\documentclass[12pt,a4paper]{amsart}
\usepackage{amsmath, amssymb, amsfonts, amsthm}
\usepackage{hyperref}
\usepackage[english]{babel}
\usepackage{amsmath}
\usepackage{amssymb}
\usepackage{url}
\usepackage{amsthm}
\usepackage{amsxtra}
\usepackage{mathrsfs}
\usepackage{fancyhdr}
\usepackage{enumerate}
\usepackage{graphicx}
\usepackage{hyperref}
\newtheorem{theorem}{Theorem}

\newtheorem{lemma}{Lemma}[section]
\newtheorem{proposition}[lemma]{Proposition}
\newtheorem{corollary}{Corollary}
\newtheorem{definition}[lemma]{Definition}

\theoremstyle{definition}
\newtheorem{remark}[lemma]{Remark}
\numberwithin{equation}{section}

\newcommand{\beq}{\begin{equation}}
\newcommand{\eeq}{\end{equation}}
\newcommand{\be}{\begin{equation*}}
\newcommand{\ee}{\end{equation*}}
\newcommand{\n}{\noindent}

\newcommand{\RE}{\mathbb R}
\newcommand{\erre}{\mathbb R}
\newcommand{\CO}{\mathbb C}
\newcommand{\NA}{\mathbb N}

\newcommand{\GG}{\mathcal{G}}

\newcommand{\supp}{\operatorname{supp}\,}
\newcommand{\Ran}{\operatorname{Ran}\,}

\newcommand{\lf}{\left}
\newcommand{\ri}{\right}
\newcommand{\ve}{\varepsilon}
\newcommand{\al}{\alpha}

\newcommand{\ga}{\gamma}

\newcommand{\la}{\lambda}
\newcommand{\de}{\delta}
\newcommand{\De}{\Delta}
\newcommand{\ci}{\mathbb{C}}

\newcommand{\ome}{\omega}
\DeclareMathOperator{\sech}{sech}
\DeclareMathOperator{\arctanh}{arctanh}
\providecommand{\ove}[1]{\overline{#1}}

\renewcommand{\Re}{\operatorname{Re}\,}

\renewcommand{\leqslant}{\leq}
\renewcommand{\geqslant}{\geq}

\newcommand{\x}{\underline{x}}
\newcommand{\y}{\underline{y}}
\newcommand{\f}{\frac}
\newcommand{\EE}{\mathcal E}

\newcommand{\VV}{V}
\newcommand{\WW}{W}
\newcommand{\ZZ}{Z}
\newcommand{\omestar}{\tilde{\ome}}

\title[]{Constrained energy minimization and orbital stability for the NLS equation on a star graph}

\author[]{Riccardo Adami}
\address{Dipartimento di Scienze Matematiche, Politecnico di Torino,  C.so Duca degli Abruzzi 24, 10129 Torino, Italy}
\email{riccardo.adami@unimib.it}

\author[]{Claudio Cacciapuoti}
\address{Hausdorff Center for Mathematics,
 Institut f\"ur Angewandte Mathematik, Endenicher Allee, 60, 53115 Bonn, Germany}
\email{cacciapuoti@him.uni-bonn.de}

\author[]{Domenico Finco}
\address{Facolt\`a di Ingegneria, Universit\`a Telematica
Internazionale Uninettuno,  Corso Vittorio Emanuele II 39, 00186 Roma,
Italy}
\email{d.finco@uninettunouniversity.net}

\author[]{Diego Noja}
\address{Dipartimento di Matematica e Applicazioni, Universit\`a
 di Milano Bicocca,  via R. Cozzi, 53, 20125 Milano, Italy}
\email{diego.noja@unimib.it} 

\date{}

\begin{document}

%%%%%%%%%%%%%%%%%%%%%%%%%%%%%%%%%%%%%%%%%
%MAKETITLE
%%%%%%%%%%%%%%%%%%%%%%%%%%%%%%%%%%%%%%%%%
%\maketitle

%%%%%%%
%ABSTRACT
%%%%%%%
\begin{abstract}
We consider a nonlinear Schr\"odinger equation with focusing nonlinearity of power type on a star graph ${\mathcal G}$, written as $
i \partial_t \Psi (t)  =   H \Psi (t) - |\Psi (t)|^{2\mu}\Psi (t)$ , 
where $H$ is  the selfadjoint operator which defines  the linear
dynamics on the graph with an attractive $\delta$ interaction, with
strength $\alpha < 0$, at the vertex. The mass and energy functionals are conserved by the flow. 
We show that for $0<\mu<2$  the energy at fixed mass is bounded from
below and that for every mass $m$ below a critical mass $m^*$ it
attains its minimum value at a certain $\hat \Psi_m \in H^1(\GG) $,
while for $m>m^*$ there is no minimum. Moreover, the set of minimizers
has the structure ${\mathcal M}=\{e^{i\theta}\hat \Psi_m\ , \theta\in
\erre \}$. Correspondingly, for every $m<m^*$ there exists a unique
$\omega=\omega(m)$ such that the standing wave
$\hat\Psi_{\omega}e^{i\omega t} $ is orbitally stable. To prove the
above results we adapt the concentration-compactness method to the
case of a star graph. This is non trivial due to the lack of
translational symmetry of the set supporting the dynamics, i.e. the
graph. This affects in an essential way the proof and the statement of
concentration-compactness lemma and its application to minimization of
constrained energy. The existence of a mass threshold comes from the instability of
the system in the free (or Kirchhoff's) case, that in our setting
corresponds to $\al=0$. 
%{\bf ***}
\end{abstract}

\maketitle
%%%%%%%%%%%%%%%%%%%%%%%%%%%%%%%%%%%%%%%%%
%SECTION
%%%%%%%%%%%%%%%%%%%%%%%%%%%%%%%%%%%%%%%%%

\section{Introduction}
In the present paper we study the minimization of a constrained energy functional defined on a star graph and its application to existence and stability of standing waves for nonlinear
Schr\"odinger propagation with an attractive interaction at the vertex of the graph. 
% More details about the definition of the model and its well posedness
% will be given in a following section (see also \cite{[ACFN4],[ACFN3]}).

\n
We recall that in our setting a star graph $\GG$ is the union of $N$ half-lines (\emph{edges}) connected at a  single vertex;  the Hilbert space on $\GG$ is $L^2(\GG)=\bigoplus_{j=1}^N L^2(\RE^+)$. We denote the elements of $L^2(\GG)$ by capital Greek letters, while functions in $L^2(\RE^+)$ are  denoted by lowercase Greek letters. The elements of $L^2(\GG)$ can be represented as column vectors of functions in $L^2(\RE^+)$, i.e.
%%%
\[
\Psi =
\begin{pmatrix}
\psi_1 \\ \vdots \\ \psi_N
\end{pmatrix}.
\]
We shall also use the notation  $\psi_i(x)\equiv (\Psi)_i(x)\equiv \Psi(x,i)$. Notice that the set
$\GG$ has not to be thought of as embedded in $\erre^n$, so it has no
geometric properties such as angles between edges.  
When an element of $L^2 (\GG)$ evolves in time, to highlight the
dependence on the time parameter $t$, we use both the notation
$\Psi(t)$ and the one with  
subscript $t$, for instance $\Psi_t$.

\n
In order to define a selfadjoint operator $H_\GG$ on $\GG$ one has to
introduce operators acting on the edges and to prescribe a suitable boundary condition
at the vertex that defines  ${\mathcal D} (H_{\GG})$,  
%{\bf ***}
%On the graph $\GG$ a Schr\"odinger operator $H_{\GG}$ is defined, with a domain ${\mathcal D} (H_{\GG})$ the elements of which satisfy a boundary condition at the vertex which makes the operator selfadjoint,
 see, e.g., \cite{[KS99]}.  A metric graph equipped with a dynamics
 associated to a Hamiltonian of the form of $H_\GG$ is  called {\it
   quantum graph}. On a quantum graph one can consider the dynamics
 defined by the abstract Schr\"odinger equation given by  
\[
%\label{linschr}
i \partial_t \Psi (t) \ = \  H_{\mathcal G} \Psi (t)\ ,\ \ \Psi\in {\mathcal D} (H_{\GG}).
\]
From a formal point of view the previous equation is equivalent to a system of $N$ Schr\"odinger equations on the half-line coupled through the boundary condition at the vertex. 

\n
Of course the graph could be more general than a star graph, with several (possibly infinite) vertices, bounded edges connecting them (sometimes called {\it bonds} as suggested from chemistry applications) 
and unbounded edges, as in the present case of star graphs or in the interesting case of trees with the last generation of edges of infinite length.

\n
The analysis of linear dispersive equations  on graphs, in particular of the Schr\"odinger equation, is a quite developed subject with a wide range of applications from chemistry and nanotechnology to quantum chaos. We refer to \cite{[BCFK06],  [BEH],  EKKST08, [Kuc04], [Kuc05]} for further information and bibliography. 

\n On the contrary, the study of nonlinear equations on networks is in general a subject at
its beginnings. Some results concerning nonlinear PDE's on
graphs are given in \cite{[CMu]} for reaction-diffusion
equations (see references therein) and in the recent paper \cite{[CMS]} for the Hamilton-Jacobi
equation (with reference to previous work on fully nonlinear
equations). As regards semilinear dispersive equations we mention the
preliminary work on NLS in the cubic 
case in \cite{[CTH]}, and for a different nonlinear dispersive
equation related to long water waves, the BBM equation, the results given
in \cite{[BC]}.\par\noindent
%{\bf ***} 
One way to define a nonlinear Schr\"odinger dynamics (NLS) on a graph,
mimicking the linear case, consists in prescribing the NLS on every
single edge and requiring its strong solution  
to satisfy a boundary condition at the vertex at every time,
i.e. imposing the solution to remain at any time in the domain of the
generator of the linear dynamics. In strong formulation, one obtains the
equation 
\[
%\label{gnls}
i \partial_t \Psi (t) \ = \  H_{\mathcal G} \Psi (t) + G(\Psi
(t)),\ \ \Psi (t)\in {\mathcal D} (H_{\mathcal G})\ ,
\]
where the nonlinearity $G=(G_1,\cdots, G_N):\CO^N\rightarrow \CO^N$
acts ``componentwise" as $G_i(\zeta)=g(|\zeta_i|)\zeta_i$ for a
suitable $g:\erre^+\rightarrow \erre$ and
$\zeta=(\zeta_1,\cdots,\zeta_N)\in \CO^N$. More general nonlinearities
of nonlocal type which couple different edges are possible at a
mathematical level, but they seem to be less interesting from the physical point of view. 
%{\bf ***}

\n
The analysis of nonlinear propagation on graphs, as in the more standard case of $\erre^n$, proceeds along two main lines of development: the study of dispersive and scattering behavior (see \cite{[ACFN1]} and reference therein; see also \cite{[BI11]} for relevant work about dispersion on trees) and the study of bound states (see \cite{[ACFN2], [ACFN4], [ACFN3]} and reference therein). In this paper we concentrate on this last item. 
%To this end and to give definite results, we complete the description of the model, specifying the nonlinearity and the interaction at the vertex of the star graph, which means to give the function $g$ and the selfadjoint operator $H_{\GG}$. 
We shall focus on a concrete model and not on a general class specifying the nonlinearity and the interaction at the vertex of the star graph, which means to give the function $g$ and the selfadjoint operator $H_{\GG}$. 
%{\bf ***}
Concerning the first, we treat a power nonlinearity of focusing type,
i.e. $g(z)=-|z|^{2\mu},\ \mu >0 \ .$  This choice has two main
reasons. It corresponds to the most usual models considered in the
physical applications, and moreover it allows to have some explicit
and quantitative estimates needed in the proofs of our results which
could be difficult to obtain for general nonlinearities. 

\n
To motivate the choice of the linear part $H_{\GG}$ we begin to remark
that the meaning of the boundary condition is to describe suitable
local interactions occurring between different components of the
wavefunction on different edges. For example, one could be interested
in describing the effect of the presence of a localized potential well
at the vertex. This corresponds in the linear case to a confining
potential admitting one or more bound states. In the case of a NLS on
the line or more generally on $\erre^n$, the presence of a negative
potential entails the existence of trapped solitons sitting around the
minima of the potential well. These trapped solitons, of the form
$\Psi(t)=\Psi_{\omega} e^{i\omega t}$ where $\omega$ belongs to some
subset of the real line, are usually called standing waves, and are
studied for example in \cite{[GS], [GSS], [GSS2], [GNT], [W1]}, to which
we refer for information and further references concerning their
existence, variational properties, orbital and  asymptotic
stability. Here we address the analogous problem in the context of
star graphs. To fix the model we consider the so called $\delta$
vertex, which is one of the most common in the applications to quantum
graphs. 

\n
We introduce preliminarily some notations and define several functional spaces on the graph. 

\n
The norm of $L^2$-functions on $\GG$ is naturally defined by
$$
\| \Psi \|^2_{L^2 (\GG)} : = \sum_{j=1}^N \| \psi_j \|^2_{L^2
  (\erre^+)} .
$$
From now on for the $L^2$-norm on the graph we
drop the subscript and simply write $\| \cdot \|$. Accordingly, we
denote by $(\cdot,\cdot)$ the scalar product in $L^2(\GG)$. 

\n
Analogously, given $1 \leqslant r \leqslant \infty$,
we define the space $L^r (\GG)$ as the set of functions
on the graph whose components are elements of the space $L^r (\erre^+)$,
and the norm is correspondingly defined by
\begin{equation*}
\big\|\Psi\big\|_{r}^{r}
=\sum_{j=1}^N\|\psi_j\|_{L^r(\RE^+)}^{r},
\ 1 \leqslant r < \infty, \qquad 
\big\|\Psi\big\|_{\infty}= \max_{1 \leqslant j \leqslant N}\|\psi_j\|_{L^\infty(\RE^+)} .
\end{equation*} 

\n
Besides, we need to introduce the spaces
\[
H^1(\GG) \equiv  \bigoplus_{j=1}^N   H^1(\erre^+)  \qquad
H^2(\GG) \equiv \bigoplus_{j=1}^N   H^2(\erre^+) 
, 
\]
equipped with the norms
\be \label{sobbo}
\| \Psi \|_{H^1}^2 \ = \ \sum_{i=1}^N \| \psi_i \|_{H^1(\erre^+)}^2,
\qquad
\| \Psi \|_{H^2}^2 \ = \ \sum_{i=1}^N \| \psi_i \|_{H^2(\erre^+)}^2.
\ee
Notice that there is a slight abuse in the denominations ${H^i}(\GG)$ for the above spaces, because their elements have no Sobolev regularity at the vertex.\par\noindent However they have boundary values on each edge, and we denote without comment the notation $\psi(0^+)=\psi(0)$ for every $\psi\in H^i(\erre^+)\ , i =1,2\ .$ 
In the following, whenever a functional norm refers to a function defined on the graph,
we omit the symbol $\GG$.

\n
We denote by $H$ the Hamiltonian with $\delta$ coupling in the vertex of strength $\alpha$, where $\al\in\RE$. It is defined as the operator in $L^2$ with  domain 
\[
%\label{e:Hdom}
{\mathcal D} (H) :=  \left\{\Psi \in H^2 \text{ s.t. }
\psi_1 (0) =  \ldots =\psi_N (0), \,  \sum_{k=1}^N \psi_k '(0)= \alpha \psi_1 (0)
\right\}.
\]
and action 
\[
%\label{e:Hact}
 H\Psi = \begin{pmatrix}
          -\psi_1''\\
	  \vdots\\
	  -\psi_N'' 
         \end{pmatrix}.
\]
In the present paper we will consider only the case of attractive $\delta$ interaction, i.e. $\alpha < 0$. Sometimes to make explicit the fact that $\al<0$ we set $\al=-|\al|$.

\n
It is well known that the operator $H$ is a selfadjoint operator on $L^2$, see, e.g., \cite{[KS99]}. Moreover for $\alpha<0$ the operator $H$ admits a single bound state associated to the eigenvalue $-{\alpha^2}/{N^2}$, in this sense the $\delta$ interaction can be considered as a singular potential well placed at the vertex.

\n
The definition of $H$ and its scope is analogous to the case of the attractive $\delta$ potential on the line, widely used in theoretical and applied physics to describe situations of strongly localized interactions such as trapping defects in a elsewhere homogeneous medium. This 
%on one hand 
is justified in view of the fact that the operator $H$ is a norm resolvent limit of  regular Schr\"odinger operators on the star graph with regular potentials $V_\epsilon$ scaling as a $\delta$-like sequence picked at the vertex (see, e.g., \cite{[Exn]})
% , and on the other hand it is useful because, as in the case of the line, also for a $\delta$  vertex on a graph one can often get exact results or solvable models. 

\n
This ends the construction and mathematical justification of the model, which is finally described by the equation
\beq
\label{schrod}
i \partial_t \Psi (t) \ = \  H \Psi (t) - |\Psi (t)|^{2\mu}\Psi (t),\ \ \Psi\in {\mathcal D} (H).
\eeq
\par\noindent
From the point of view of physical applications the problem described by the above equation is interesting in relation to the so called Y-junctions or beam splitters in the study of Bose-Einstein condensates (see \cite{[TOD]}). Other problems related to nonlinear Schr\"odinger propagation on graphs are treated in \cite{[GSD], [Miro], [Sob]}, and more generally there is a growing interest in  nonlinear propagation on networks, both in nonlinear optics and in Bose condensates, which are the main fields of application of the  NLS. 

\n
From the mathematical point of view,  several results on the  nonlinear model \eqref{schrod}  were given in a series of papers (\cite{[ACFN1], [ACFN2], [ACFN4], [ACFN3]}). In particular we refer to the work \cite{[ACFN3]} which is a companion to the present one, where a variational study of the standing waves and their orbital stability is performed according to the Grillakis-Shatah-Strauss method (\cite{[GSS],[GSS2]}). While in \cite{[ACFN3]} the interesting functional is the action, minimized on the Nehari manifold, here we minimize the energy at constant norm following the Cazenave-Lions approach to orbital stability, see \cite{[CL]} (see also  \cite{Caz03, Caz06}).  
As it is shown elsewhere (see \cite{[ACFN3]}) the dynamical system \eqref{schrod} has two conserved quantities, the mass 
\begin{equation}
\label{e:mass}
M[\Psi]=\|\Psi\|^2 
\end{equation}
and the energy $E$, which in our case reads
\begin{equation}
\label{e:nrg}
E[\Psi]=\frac{1}{2}\| \Psi' \|^2 - \frac{1}{2\mu+2} \| \Psi \|_{2 \mu+ 2}^{2\mu+2} +
\frac{\alpha}{2}|\psi_1(0)|^2 .
%= E^0 -  \frac{1}{2\mu+2} \| \Psi \|_{2 \mu+ 2}^{2\mu+2}\ .
\end{equation}
The energy domain $\mathcal{E}$ coincides with the  domain of the quadratic form associated to the linear generator $H$, consisting of $H^1$ functions on every edge with continuity at the vertex
\be
 {\mathcal E} := \left\{\Psi\in H^1 \text{ s.t.}\
\psi_1(0)=\cdots=\psi_N(0)\right\} .
\ee
On this domain we show that the energy $E[\Psi]$ is bounded from below if the mass $\|\Psi\|^2$ is fixed.

\n
We are then interested in characterizing the ground state of this system. By ground state we mean the minimizer $\hat\Psi$ (if existing) of the energy $E$ in $\mathcal{E}$ constrained to the manifold of the states with fixed mass $m$.
%:\footnote{La formula sotto secondo me si pu\`o omettere. Comunque non indicherei  l'inf  con $E(\alpha,m)$. Nel teorema sotto usiamo $-\nu$}
%\[
%%\label{minimization}
%E(\alpha,m) \ : = \ 
%\inf_{\substack{\Psi \in {\mathcal E}\\\|\Psi\|^2=m}}  E[\Psi] = E[\hat\Psi] .
%\]

\n
As noticed before, the classical method which allows to treat this
kind of problems is  the concentration-compactness principle of
P.-L. Lions with its application to the NLS given in
\cite{[CL]}). A study of ground states for NLS on the line with several kind of defects (including the $\delta$ potential) making use of a concentration compactness  is given in \cite {[ANV12]}. Nevertheless, the present situation needs some non
trivial modifications of the method, due to the fact that a graph, and
in particular a star graph, does not enjoy translational symmetry, nor
other kinds of symmetry needed to apply concentration-compactness in
its direct form (see \cite{[TF]} for a very general presentation and
applications of the method). We will adapt the
concentration-compactness lemma (as given in \cite[Ch. 1 and 8]{Caz03}
and also in \cite{Caz06}, which we will take as reference formulation in
the course of our treatment) modifying the statement and the proof to
draw our main conclusions on the minimum problem we are interested
in. For more extended discussion on the novelties of this approach, we refer to Section 3. Using the concentration-compactness lemma we prove the following result which states the existence of the
solution  of the  constrained minimization problem for small enough
mass. 
%%%%%%%
%THEOREM
%%%%%%%
\begin{theorem}
\label{t:prob1}
Let $m^\ast$ be defined by
\beq
\label{mstar}
m^\ast=2  \f{(\mu+1)^{1/\mu} }{\mu} \, \lf( \f{|\al|}{N}\ri)^{\f{2-\mu}{\mu}  } \int_0^1 (1-t^2)^{ \f{1}{\mu}-1}\, dt.
\eeq
Let $\al <0 $ and assume $m\leq m^\ast$ if $0<\mu < 2$ and  $m < \min\{{m^\ast , \frac{\pi \sqrt{3} N}{4} }\}$ if $\mu=2$ and set
%Assume $0<\mu \leq 2$, $\al<0$ and  $m\leq m^\ast$  if $ 0 < \mu <2$ while in the critical case we assume $m < \min{m^\ast , \frac{\pi \sqrt{3} N}{4} }$ and set
%Assume $0<\mu<2$, $\al<0$ and  $m<m^\ast$ and set 
\[
-\nu = \inf \{E[\Psi] \textrm{ s.t. } \Psi\in\mathcal{E}\,,\; M[\Psi] = m\}\,.
\]
Then $0<\nu<\infty$ and there exists $\hat\Psi $ such that $M[\hat\Psi] = m$ and $E[\hat\Psi]=-\nu$.
\end{theorem}

\n
By the phase invariance of equation \eqref{schrod} one has  that the family of ground states is given by
\[
%\label{minimizers}
{\mathcal M}=\{e^{i\theta}\hat \Psi\ , \theta\in \erre \}\ .
\]
The explicit expression of $\hat \Psi$ can be given. To this end, let
us recall several results from \cite{[ACFN4]} and \cite{[ACFN3]}. For
any $\ome>0$, we label the soliton profile on the real line as 
\beq
\label{soliton}
\phi_\ome (x) = [ (\mu + 1) \ome]^{\frac{1}{2\mu}} \sech^{\frac{1}{\mu}} (\mu \sqrt{\ome} x).
\eeq
%It satisfies the  equation
%\beq 
%\label{soliteq}
%-\phi_\ome '' -|\phi_\ome |^{2\mu} \phi_\ome = -\ome \phi_\ome.
%\eeq
For any $\al<0$, $j=0, ...,\left[\frac{N-1}{2}\right]$ ($[x]$ denoting
the integer part of $x$) and $\ome > \frac{\alpha^2}{(N-2j)^2}$ we
define $\Psi_{\ome,j}$ as   
\beq
\lf(\Psi_{\omega,j}\ri)(x,i) = 
\begin{cases}
\phi_\ome(x-a_j ) & i=1,\ldots ,j \\
\phi_\ome(x+a_j) & i=j+1, \ldots, N
\end{cases}
\label{states1}
\eeq 
with 
\beq 
a_j = \f{1}{\mu \sqrt{\ome}} \arctanh
\lf(\f{|\al|}{(N-2j)\sqrt{\ome}} \ri) .
\label{states2} 
\eeq 
The functions  $\Psi_{\ome,j}\in \mathcal{D}(H)$  and are solutions of the stationary equation
\begin{equation}
\label{stationary}
H\Psi_{\omega} - |\Psi_{\omega}|^{2\mu} \Psi_{\omega} = -\ome \Psi_{\omega}.
\end{equation}
We say that $\Psi_{\ome,j}$ has a ``bump'' (resp. a ``tail'') on the edge $i$ if $\lf(\Psi_{\omega,j}\ri)(x,i)$ is of the form $\phi_\ome(x-a_j )$ (resp. $\phi_\ome(x+a_j )$).
The index $j$ in $\Psi_{\ome,j}$ denotes the number of bumps of the state $\Psi_{\ome,j}$. For this reason,  we refer to the stationary state $\Psi_{\ome,0}$ as the ``$N$-tail state''.  
 We remark that the $N$-tail state is the only symmetric (i.e. invariant under permutation of the edges)  solution of equation \eqref{stationary}. For $j\geq 1$ there are $\binom{N}{j}$ distinct solutions obtained by formulae \eqref{states1} and \eqref{states2} by positioning the bumps on the edges in all the possible ways.  For
instance, if $N=3$ then there are two stationary states, a three-tail state
and a two-tail/one-bump state. They are shown in figure \ref{figu1}. 

\mbox{}

\begin{figure}[h!] 
\centering
\includegraphics[scale=0.50]{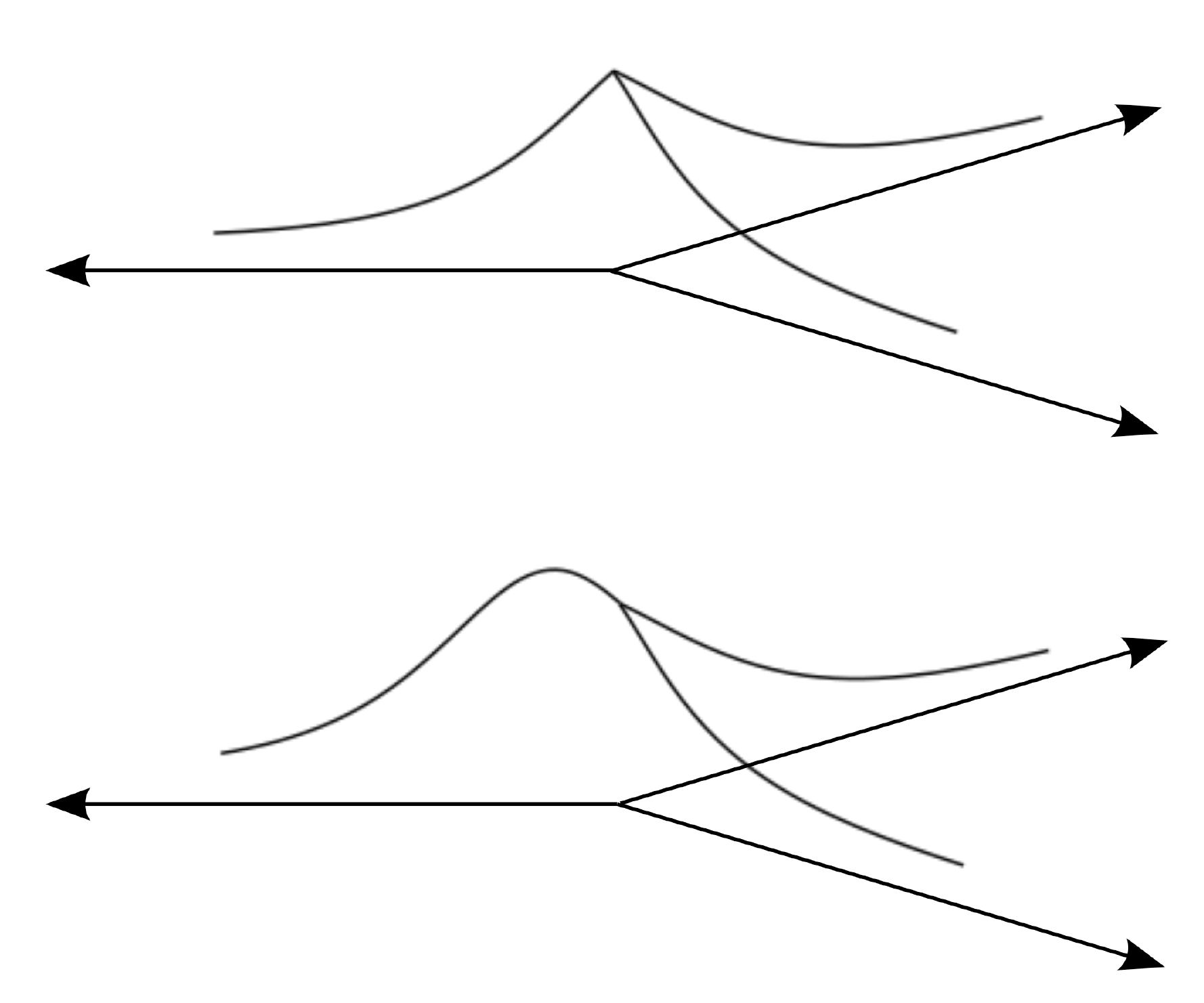}
\caption{Stationary states for $N=3$, $\alpha<0$\ .}
\label{figu1}
\end{figure}
\n
 
\noindent

%%%%%%%
%THEOREM
%%%%%%%
\begin{theorem}
\label{t:min}
Let $\al <0 $ and assume $m\leq m^\ast$ if $0<\mu < 2$ and  $m < \min\{m^\ast , \frac{\pi \sqrt{3} N}{4} \}$ if $\mu=2$; 
%Assume $0<\mu \leq 2$, $\al<0$ and  $m<m^\ast$  if $ 0 < \mu <2$ while in the critical case we assume $m < \min{m^\ast , \frac{\pi \sqrt{3} N}{4} } $
then the minimizer $\hat \Psi$ coincides with the $N$-tail state
defined by $\Psi_{\omega_0,0} $ where $\ome_0$ is chosen such that
$M[\Psi_{\omega_0,0}]=m$. 
\end{theorem}

\n
Since the minimizer $\hat \Psi$ is a stationary state, in order to prove  Th. \ref{t:min} 
it is sufficient to show that $\Psi_{\omega_0,0} $ has minimum energy among the set of stationary states, which is finite. 
In facts in Section \ref{sec5} we shall prove a more detailed statement; the energies of the stationary states, with  frequencies $\ome_j$ such that $M[\Psi_{\ome_j,j}]=m$,  
are increasing in $j$, i.e. they can be ordered in the number of the
bumps, see Lem. \ref{l:nrgord}. Notice that the bounds on thresholds
in $m$ are different in the critical and subcritical case. More
remarks on this are given in Section \ref{sec5}. Notice that as a
consequence we have that the ground state of the system is the only
stationary state which is symmetrical with respect to permutation of
edges. 

\n
Finally, making use of the classical argument of Cazenave and Lions \cite{[CL]}, from mass and energy conservation laws, convergence of the minimizing sequences  and uniqueness of the ground state up to phase shift shown in Th. \ref{t:prob1} and Th. \ref{t:min}, the orbital stability of the ground state follows. A detailed proof will not be given, being straightforward extension of the previous outline.
\begin{corollary}\label{c:orbstab}
Let $\al <0 $ and assume $m\leq m^\ast$ if $0<\mu < 2$ and  $m < \min\{m^\ast , \frac{\pi \sqrt{3} N}{4} \}$ if $\mu=2$; then $\Psi_{\omega_0,0}$ is orbitally stable.
\end{corollary}

\n
The paper is organized as follows. 
In Section \ref{sec2} we recall several known results which will be
needed in the proof of Th. \ref{t:prob1}. In Section \ref{sec3} we
prove the concentration-compactness lemma for star graphs. Section
\ref{sec4} is devoted to the proof of Th. \ref{t:prob1}. In Section
\ref{sec5} we analyze the frequency and energy of stationary states on
the manifold of constant mass and prove Th. \ref{t:min}.

%%%%%%
%SECTION
%%%%%%
\section{Preliminaries} \label{sec2}

In this section we fix some notation and recall several results mostly
taken from \cite{[ACFN3]}.  We shall denote generic positive constants
by $c$, in the proof the value of $c$ will not be specified and can
change from line to line. The dual of $\EE$ will be denoted by
$\EE^\star$. 
We shall denote the points of the star graph  by $\underline{x}\equiv (x, j)$ with $x \in \erre^+$ and $j \in \{ 1, \ldots , N \}$. 
%%%%%%%%%%%%
%SUBSECTION
%%%%%%%%%%%%
\subsection{Well-posedness} 
We recall that equation \eqref{schrod} can be understood in the weak form given by 
\begin{equation}
\label{intform1}
\Psi(t) \ = \ e^{-iH t} \Psi_0 - i \int_0^t e^{-i
  H (t-s)}|\Psi(s)|^{2\mu} \Psi(s) \, ds
\end{equation}
with $\Psi_0 \equiv \Psi(t=0)$.

\n      
As in the standard NLS on the line, mass and energy,
Eqs. \eqref{e:mass} and \eqref{e:nrg}, are conserved by the flow, see
Prop. 2.2 in \cite{[ACFN3]}. Moreover,   if  $0<\mu<2$, then the
equation \eqref{intform1} is well posed in the energy domain and the
solution is global, see Cor. 2.1 in \cite{[ACFN3]}. More precisely we
have 
\n

\begin{proposition}[Local well-posedness in $\EE$] \mbox{ }
%\label{loch2}
\n Let $\mu >0$. For any $\Psi_0 \in \EE$, there exists $T > 0$ such that the
equation \eqref{intform1} has a unique solution $\Psi \in C^0 ([0,T),
\EE )
\cap C^1 ([0,T), \EE^\star)$. Moreover, Eq. \eqref{intform1} has a maximal solution $\Psi^{\rm{max}}$
defined on an interval of the form $[0, T^\star)$, and the following ``blow-up
alternative''
holds: either $T^\star = \infty$ or
$$
\lim_{t \to T^\star} \| \Psi_t^{\rm{max}} \|_{\EE}
\ = \ + \infty,
$$
where we denoted by $\Psi_t^{\rm{max}}$ the function $\Psi^{\rm{max}}$ evaluated at time $t$.
\end{proposition}

\begin{proposition}[Conservation laws] \mbox{ }
\n Let $\mu>0$. For any solution $\Psi \in C^0 ([0,T), \EE)
\cap C^1 ([0,T), \EE^\star)$ to
the problem \eqref{intform1}, the following conservation laws hold at
any time $t$:
\begin{equation*}
%\label{conslaws}
M[\Psi_t ] \ = M[ \Psi_0 ], \qquad
E[ \Psi_t ] \ = \ E[ \Psi_0 ].
\end{equation*}
\end{proposition}
\begin{corollary}[Global well posedness] \mbox{ }
\n Let $0<\mu<2$. For any $\Psi_0 \in \EE$,  the
equation \eqref{intform1} has a unique solution $\Psi \in C^0 ([0,\infty),
\EE )
\cap C^1 ([0,\infty), \EE^\star)$. 
\end{corollary}

%%%%%%%%%%%%
%SUBSECTION
%%%%%%%%%%%%
\subsection{Kirchhoff coupling} 
The vertex coupling associated to $\al=0$, is  usually called \emph{free} (on the line the interaction disappears) or \emph{Kirchhoff} coupling and plays a distinguished role. For this reason we shall denote by $H^0$ the corresponding operator defined by
\[
{\mathcal D} (H^0): =
\{ \Psi \in H^2 \text{ s.t. } \,
\psi_1 (0) = \ldots = \psi_N (0), \, \sum_{i=1}^N \psi_i ' (0) = 0 \}
%\label{fujidelta}
\]
\[
H^0 \Psi =
\begin{pmatrix}
-\psi_1'' \\ \vdots \\ -\psi_N''
\end{pmatrix}.
\]
We also define the corresponding energy functional 
\beq
\label{e:nrgK}
E^0[\Psi]=\frac{1}{2}\| \Psi' \| - \frac{1}{2\mu+2} \| \Psi \|_{2 \mu+ 2}^{2\mu+2} \eeq
with energy domain $\mathcal{D}(E^0)= \mathcal{E}$.
%%%%%%%%%%%%%%%%%%%%%%%%%%%%%%%%%%%%%%%%%%%%%%%%%%%%%%%%%%%%%
%LINEAR ENERGY FUNCTIONALS
%%%%%%%%%%%%%%%%%%%%%%%%%%%%%%%%%%%%%%%%%%%%%%%%%%%%%%%%%%%%%
% \n
% The quadratic form $E^{lin}$
% associated to $H$ is defined on the energy space, i.e.  $\mathcal{D}(E^{lin})\equiv\EE$
% % \begin{equation*}
% % \EE \equiv{\mathcal D} ( E^{lin} ) = \{ \Psi \in H^1 (\GG)
% % \text{ s.t. } \,
% % \psi_1 (0) = \ldots = \psi_N (0) \}
% % %\label{toukyoudelta}
% % \end{equation*}
% and is given by
% \begin{equation*}
% E^{lin} [\Psi ] = 
% \f{1}{2} \| \Psi '\|^2 + \f{\alpha}{2} | \psi_1 (0) |^2 = 
% \f{1}{2} \sum_{i=1}^N \int_0^{+\infty}
% |\psi_i ' (x) |^2 \,dx\, + \f{\alpha}{2} | \psi_1 (0) |^2.
% \end{equation*}
% In the same way, the quadratic form $E^{0,lin}$ associated to $H^0$ is defined on
% the  same subspace, that is ${\mathcal D} ( E^{0, lin} ) = \EE$,
% and reads
% \begin{equation*}
%  E^{0,lin} [\Psi ] = \f{1}{2} \| \Psi '\|^2 =\f{1}{2} \sum_{i=1}^N \int_0^{+\infty} |\psi_i ' (x) |^2 \,dx\,.
% \end{equation*}
%%%%%%%%%%%%
%SUBSECTION
%%%%%%%%%%%%
\subsection{Gagliardo-Nirenberg inequalities} 
We shall use a version of Gagliardo-Nirenberg inequalities on the star
graph. The following proposition is a direct consequence of the
Gagliardo-Nirenberg inequalities on the half-line see, e.g.,
\cite[I.31]{MPF91}. 
%%%%%%%%%
%PROPOSITION
%%%%%%%%%
\begin{proposition}[Gagliardo-Nirenberg Inequality]
\label{p:GN}
Let  $2\leq q\leq+\infty$, $1\leq p\leq q$ and set  $a= \frac{\frac{1}{p}-\frac{1}{q}}{\frac{1}{2}+\frac{1}{p}}$, then for any $\Psi\in H^1$
\[
\|\Psi\|_{q} \leq c \|\Psi'\|^a \|\Psi\|_{p}^{1-a}\,. 
\]
\end{proposition} 
%%%%%%%%%%%%
%SUBSECTION
%%%%%%%%%%%%
\subsection{Symmetric rearrangements} 
\n
Here we recall the basic properties of symmetric rearrangements on a
star graph introduced in \cite{[ACFN3]}.  
For a given function $\Phi:\GG \to \ci^N$ one introduces the
rearranged function $\Phi^*:\GG \to \erre^N$. The function $\Phi^*$ is
positive, symmetric, non increasing and   is constructed in such a way 
that it is equimisurable w.r.t. $\Phi$, that is, the level sets of
$|\Phi|$ and $\Phi^*$ have the same measure. This is sufficient to
prove that all the $L^p(\GG)$ norms are conserved by the
rearrangement. The comparison of the kinetic energy of $\Phi$ and
$\Phi^*$ is more delicate. On the real line  the P\'olya-Szeg\H o
inequality shows that the kinetic energy does not increase. This is no
longer true for a star graph where a constant $N/2$ appears, see
Prop. \ref{polya} below.

\begin{definition}[Symmetric rearrangement]
Given $\Psi: \GG \to \ci^N$, let $\la : \GG \to \erre$ be given by  $\la (s) = | \{ |\Psi| \geqslant s\} |$ (which is the measure of the set $\{\underline x \textrm{ s.t. } |\Psi(\underline x)| \geqslant s\}$) and  $g: \erre^+ \to \erre^+$  be  $g(t) = \sup \{ s | \, \la(s) > N t\}$. The symmetric rearrangement $\Psi^\ast$ of $\Psi$ is defined by $\Psi^*=(\psi^*_1,...,\psi^*_N)^T$ with  
\[
\psi^*_1(x)= . . . = \psi^*_N(x) = g (x) .
\]
\end{definition}
The main properties of $\Psi^{\ast}$ are the following:
\begin{proposition} \label{lp}
The symmetric rearrangement $\Psi^{\ast}$ is  positive, symmetric and
non increasing. Moreover,  $\| \Psi^{\ast}\|_{p  } =\| \Psi\|_{p  }$.
\end{proposition}
%For a star graph the P\'olya-Szeg\H{o} inequality takes a different form compared to the case of the real line.
\begin{proposition}[P\'olya-Szeg\H{o} inequality for star graphs]   \label{polya}
Assume that $\Psi \in H^1 $. Then $\Psi^\ast \in H^1 $ and
\[
% \label{chiacchere}
\|{\Psi^{\ast}}'\|  \leqslant  \f{N}{2}    \|\Psi'\|. 
\]
\end{proposition}
%%%%%%%%%%%%%
%%SUBSECTION
%%%%%%%%%%%%%
%\subsection{Stationary states\label{ss:statstat}}
%For any $\ome>0$, we label the soliton profile on the real line as
%\beq
%\label{soliton}
%\phi_\ome (x) = [ (\mu + 1) \ome]^{\frac{1}{2\mu}} \sech^{\frac{1}{\mu}} (\mu \sqrt{\ome} x).
%\eeq
%It satisfies the  equation
%\beq 
%\label{soliteq}
%-\phi_\ome '' -|\phi_\ome |^{2\mu} \phi_\ome = -\ome \phi_\ome.
%\eeq
%Then for any $\al<0$, $j=0, ...,\left[\frac{N-1}{2}\right]$ ($[x]$ denoting the integer part of $x$) and $\ome > \frac{\alpha^2}{(N-2j)^2}$ we define $\Psi_{\ome,j}$ as  
%\beq
%\lf(\Psi_{\omega,j}\ri)(x,i) = 
%\begin{cases}
%\phi_\ome(x-a_j ) & i=1,\ldots ,j \\
%\phi_\ome(x+a_j) & i=j+1, \ldots, N
%\end{cases}
%\label{states1}
%\eeq 
%with 
%\beq 
%a_j = \f{1}{\mu \sqrt{\ome}} \arctanh
%\lf(\f{\al}{(2j-N)\sqrt{\ome}} \ri) .
%\label{states2} 
%\eeq 
%The functions  $\Psi_{\ome,j}\in \mathcal{D}(H)$  and are solutions of the stationary equation
%\[
%H\Psi_{\omega} - |\Psi_{\omega}|^{2\mu} \Psi_{\omega} = -\ome \Psi_{\omega}.
%\]
%In what follows we will refer to the stationary state $\Psi_{\ome,0}$ as the ``$N$-tail state''. 
%%%%%%%%%%%%
%SUBSECTION
%%%%%%%%%%%%
\subsection{Mass and energy on the half-line and on the line}
For later convenience we introduce also the unperturbed energy and
mass functional for functions belonging to $H^1 (\erre^+)$ and $H^1
(\erre)$. 
For the half-line we denote the functionals by $M_{\erre^+}$ and
$E_{\erre^+}$, respectively. They are defined by
\begin{align*}
& M_{\erre^+} [\psi] = \| \psi \|^2_{L^2(\erre^+)} \\
& E_{\erre^+}^0 [\psi] = \f 1 2 \| \psi' \|^2_{L^2(\erre^+)} -
  \frac{1}{2\mu + 2} \| \psi \|^{2\mu + 2}_{L^{2\mu +2} (\erre^+)}. 
\end{align*}
For the line we denote the mass end energy functionals by $M_{\erre}$ and
$E_{\erre}$, respectively. They are defined by
\begin{align*}
& M_{\erre} [\psi] = \| \psi \|^2_{L^2(\erre)} \\
& E_{\erre}^0 [\psi] = \f 1 2 \| \psi' \|^2_{L^2(\erre)} - \frac{1}{2\mu + 2} \| \psi \|^{2\mu + 2}_{L^{2\mu +2} (\erre)}.
\end{align*}

\n
Using the definition \eqref{soliton} and a change of variable,  one
obtains the following formulas:
\begin{align}
&\int_0^\infty |\phi_\ome (x+\xi )|^2 dx \quad=
  \frac{(\mu+1)^{\frac{1}{\mu} }}{\mu} \ome^{\frac{1}{\mu}
    -\frac{1}{2} }  
\int^1_{\tanh (\xi \mu\sqrt{\ome})} (1-t^2)^{\frac{1}{\mu} -1} dt \label{formula1} \\
&\int_0^\infty |\phi_\ome (x+\xi )|^{2\mu+2} dx =
\frac{(\mu+1)^{1+\frac{1}{\mu} }}{\mu} \ome^{\frac{1}{\mu}
  +\frac{1}{2} }  
\int^1_{\tanh (\xi \mu\sqrt{\ome})} (1-t^2)^{\frac{1}{\mu} } dt \label{formula2} .
\end{align}
The mass and energy functional evaluated on the soliton are given by
\begin{align}
\label{MR}
& M_\erre [\phi_\ome] = 2 M_{\erre^+} [\phi_\ome] = 2  \f{(\mu+1)^{\f 1 \mu }  }{\mu} \ome^{ \f 1 \mu - \f 1 2} \int_0^1 (1-t^2)^{\f 1 \mu -1} dt, \\
\label{ER}
& E^0_\erre [\phi_\ome ]=2 E^0_{\erre^+} [\phi_\ome ]=  -  \f{(\mu+1)^{\f 1 \mu }  }{\mu}   \f{2-\mu}{2+\mu}   \ome^{ \f 1 \mu + \f 1 2} \int_0^1 (1-t^2)^{\f 1 \mu -1} dt,
\end{align}
where we used the identity
\begin{equation}
\lf( \frac{1}{2} + \frac{1}{\mu} \ri) \int^1_{b} (1-t^2)^{\frac{1}{\mu} } dt=
-\f{b}{2} (1-b^2)^{\frac{1}{\mu}} +
\frac{1}{\mu} \int^1_{b} (1-t^2)^{\frac{1}{\mu}-1 } dt.
\label{formula3}
\end{equation}
\n
It is well known that the function $\phi_\ome$ minimizes $E_{\erre}$
at fixed mass. More precisely, choose $\ove \ome$ such that
$M_\RE[\phi_{\ove \ome}] =m$,
then $\phi_{\ove \ome}$ is a minimizer of the  problem
\be
\inf_{\stackrel{\psi\in H^1 (\erre)}{  M_\erre[\psi] =m }  } E_\erre^0 [\psi] .
\ee

\n
This also implies that $\phi_\omega$, with $\omega$ such that
$M_{\RE^+}[\phi_\ome] =m$, is the solution to  the problem 
\[
%\label{stima}
\inf_{\stackrel{\psi\in H^1 (\erre^+)}{ M_{\erre^+}[\psi] =m}} E_{\erre^+ }[\psi].
\]
To prove the last statement, assume that  $f\in H^1(\RE^+)$ is such that $M_{\RE^+}[f] = m$ and
\[
 E^0_{\RE^+} [f] \leq  E_{\RE^+} [\phi_\ome]
\]
 where $\ome$ is chosen to satisfy $M_{\RE^+}[\phi_\ome]=m$. Then,
 denoted by $\tilde f$ the even extension of $f$,  one would obtain 
\[
 E^0_{\RE} [\tilde f] \leq  E^0_{\RE} [\phi_\ome]
\]
where $M_\RE[\tilde f] = M_\RE[\phi_\ome]=2m$. Since
$\phi_\ome$ is, up to a phase, the only minimizer of $E^0_\RE$ at fixed mass, $f$ must
be equal to $\phi_\ome$ up to a phase factor. 
% 
% *******************************************\\
% *******************************************\\
% *******************************************\\
% *******************************************\\
% *******************************************\\
% *******************************************\\
% 
% We can reduce to the minimization of a functional on the line by means of the following simple inequality
% \beq \label{stima}
% \inf_{\substack{\psi\in H^1 (\erre^+) \\ M_{\erre^+}[\psi] =m}}
% E_{\erre^+ }[\psi] \geqslant \f 1 2 
% \inf_{\substack{\psi\in H^1 (\erre) \\ M_{\erre}[\phi] =2m}}
% E_{\erre}[\phi]
% \eeq
% If \eqref{stima} would be false we could a minimizing sequence on the half line obtaining energies less than the infimum of the line.
% 
% \n
% The solution of the problem on the real line is known to be the energy of $\phi_\ome$ with a suitably chosen $\ome$ in order to satisfy the mass constrain,
% that is, 
% \beq
% \inf_{\substack{\psi\in H^1 (\erre)\\ M_{\erre[\psi]}=m}}
% E_{\erre} [ \psi] = E_{\erre} [ \phi_\ome] = 
% % \f{1}{2} \| \psi^{ \prime}  \|^2_{L^2 (\erre)} - \f{1}{2\mu + 2} \| \psi \|_{L^{2\mu + 2}(\erre)}=-
%  \f{(\mu+1)^{\f 1 \mu }  }{\mu}   \f{2-\mu}{2+\mu}   \ome_{\erre}^{ \f 1 \mu + \f 1 2} \int_0^1 (1-t^2)^{\f 1 \mu -1} dt
% \eeq
% where the frequency $   \ome_{\erre}$ satisfies
% \beq \label{retta}
% m = 2 \f{(\mu+1)^{\f 1 \mu }  }{\mu} \ome_{\erre}^{ \f 1 \mu - \f 1 2} \int_0^1 (1-t^2)^{\f 1 \mu -1} dt \\
% \eeq

%%%%%%%%%
%SECTION
%%%%%%%%%
\section{Concentration-compactness lemma}
\label{sec3}
\n
In this section we prove the 
concentration-compact\-ness lemma, that will be
the main tool in the proof of Th. \ref{t:prob1}.
\n
%The strategy is the following. First we prove the
%concentration-compactness lemma (Lem. \ref{l:cc} below). 
For any sequence $\{\Psi_n\}_{n\in \NA}$ such that $M[\Psi_n] \to m$ and $\|\Psi_n\|_{H^1}$ is
bounded, the lemma states the existence of  a
subsequence whose behavior  is decided by the concentrated mass $\tau$
(see Section \ref{sec3} for the precise definition). We distinguish
three cases: $\tau=0$, $0<\tau<m$ and $\tau=m$, corresponding
respectively to \emph{vanishing}, 
\emph{dichotomy} or \emph{compactness}, which are the usual, well known
possibilities in the standard concentration-compactness theory. We
remark that the statement of the lemma concerns  
the existence {\em of a 
subsequence only} of $\{\Psi_n\}_{n\in \NA}$ having the behavior defined by the value
of the parameter $\tau$. In other words, the lemma does not characterize all the subsequences of
$\{\Psi_n\}_{n\in\NA}$. The novel point in the extension of the theory to sequences of
functions defined on the star graph $\GG$, concerns the case of
compactness. Indeed, as in the standard case, a compact sequence can 
either remain essentially concentrated in a finite region and then
strongly 
converge, or escape towards the infinity. The lack of translational
invariance in  $\GG$ forces to distinguish these two cases,
so we say that the subsequence is \emph{convergent}
if it converges to some function $\Psi\in\mathcal E$ (case $i_1)$ of
Lem. \ref{l:cc}), and we say that 
the subsequence is \emph{runaway} if 
the %substantial part of the
subsequence carries the whole mass towards infinity along a single
edge  (case $i_2)$ in Lem. \ref{l:cc}).
\n
In the development of the concentration-compactness theory, we closely
follow the roadmap of
\cite{Caz03, Caz06}, generalizing at any step to the case of the star
graph the corresponding result of the standard theory in $\RE^n$. 

We start by
defining the distance between points of the graph,
then we introduce the concentration function and analyze its properties.

Let $\x=(x,j)$ and $\y=(y,k)$, with $j,k=1,...,N$ and $x,y\in\RE_+$, two points
of the graph and define the distance 
\[
d(\x,\y)\equiv d((x,j),(y,k)):=
\left\{
\begin{aligned}
&|x-y| \quad & \textrm{for }j=k\\
&x+y & \textrm{for } j\neq k
\end{aligned}
\right.  
\]

We denote by $B(\y,t)$ the open ball of radius $t$ and center $\y$
\[
B(\y,t):=\{\x\in\GG \textrm{ s.t. } d(\x,\y)<t\}\,,
\]
and by $\|\cdot\|_{B(\y,t)}$ the $L^2(\GG)$ norm restricted to the ball
$B(\y,t)$, i.e. set $\y=(y,k)$ then
\[
\|\Psi\|_{B(\y,t)}^2 = \int_{\{x\in\RE_+ \textrm{ s.t. } |x-y|<t\}}|\psi_k(x)|^2
dx +
\sum_{j\neq k , j=1}^N \int_{\{x\in\RE_+ \textrm{ s.t. } x+y<t\}}|\psi_j(x)|^2
dx \,.
\]
For any function $\Psi\in L^2$ and $t\geq0$ we define the concentration function
$\rho(\Psi,t)$ as
\beq
\label{e:rho}
\rho(\Psi,t) = \sup_{\y\in\GG} \|\Psi\|_{B(\y,t)}^2\,.
\eeq

In the following proposition we prove two important properties of the
concentration function: that the $\sup$ at the r.h.s. of equation \eqref{e:rho}
is indeed attained at some point of $\GG$ and the H\"older continuity of
$\rho(\Psi,\cdot)$.

%%%%%%%%%
%PROPOSITION
%%%%%%%%%
\begin{proposition}
\label{p:rho}
Let $\Psi\in L^2$ such that $\|\Psi\|>0$, then
\begin{enumerate}[i)]
\item{\label{i:rho0}}
$\rho(\Psi,\cdot)$ is
non-decreasing,  $\rho(\Psi,0)=0$, $0<\rho(\Psi,t)\leq M[\Psi]$ for $t>0$,
and $\lim_{t\to\infty} \rho(\Psi,t) =  M[\Psi]$. 
\item{\label{i:rho1}}
 There exists $\y(\Psi,t)\in\GG$ such that 
\[
\rho(\Psi,t)= \|\Psi\|_{B(\y(\Psi,t),t)}^2\,.
\]
\item
\label{i:rho2}
If $\Psi\in L^p$ for some $2\leq p\leq\infty$, then 
\beq
\label{e:ii}
|\rho(\Psi,t)-\rho(\Psi,s)| \leq c \|\Psi\|_{p}^2 |t-s|^{\frac{p-2}{p}} \qquad
\textrm{for } 2\leq p<\infty
\eeq
and 
\beq
\label{e:ii-infty}
|\rho(\Psi,t)-\rho(\Psi,s)| \leq c \|\Psi\|_{\infty}^2 |t-s| \qquad
\textrm{for } p=\infty
\eeq
for all $s,t>0$ and where $c$ is independent of $\Psi$, $s$ and $t$.
\end{enumerate}
\end{proposition}
%%%%%
%PROOF
%%%%%
\begin{proof}
Proof of \emph{\ref{i:rho0})}. This follows directly from the definition of
$\|\cdot\|_{B(\y,t)}$ and $\rho(\Psi,t)$. 
 
Proof of \emph{\ref{i:rho1})}.  Let $\{\y_n\}_{n\in \NA}$ be a sequence such
that $\lim_{n\to\infty} \|\Psi\|_{B(\y_n,t)}^2=\rho(\Psi,t)$. To prove
\emph{\ref{i:rho1})} it is enough to prove that $\{\y_n\}_{n\in \NA}$ is
bounded. Assume that $\{\y_n\}_{n\in \NA}$ is not bounded, then there exists a
subsequence $\{\y_{n_k}\}_{k\in\NA}$ such that the balls $B(\y_{n_k},t)$ and
$B(\y_{n_l},t)$ are disjoint for any $k\neq l$, and
$\|\Psi\|_{B(\y_{n_k},t)}^2\geq\rho(\Psi,t)/2$  for all $k$. This is absurd
because it would imply 
\[
\|\Psi\|^2\geq \sum_{k=1}^\infty  \|\Psi\|_{B(\y_{n_k},t)}^2 =\infty\,.
\]
Therefore $\{\y_n\}_{n\in \NA}$ is bounded and consequently has a convergent
subsequence whose limit is $\y(\Psi,t)$.

Proof of \emph{\ref{i:rho2})}. Without loss of generality we can assume $0<s\leq
t<\infty$. Then for any $\y=(y,k)\in\GG$, $B(\y,t)=B(\y,s)\cup \left(B(\y,t)\backslash
B(\y,s)\right)$ and
$\|\Psi\|_{B(\y,t)}^2=\|\Psi\|_{B(\y,s)}^2+\|\Psi\|_{B(\y,t)\backslash
B(\y,s)}^2$.
% where 
% \beq
% B(\y,t)\backslash B(\y,s) =\{\x\in\GG \textrm{ s.t. }  s\leq d(\x,\y)<t\}
% \eeq
% and 
% \beq
% \label{e:B-B}
% \|\Psi\|_{B(\y,t)\backslash B(\y,s)}^2 = 
% \int_{\{x\in\RE_+\textrm{ s.t. } s\leq |x-y|<t\}}|\psi_k(x)|^2 dx + \sum_{j\neq
% k, j=1}^N \int_{\{x\in\RE_+\textrm{ s.t. } s\leq x+y <t\}}|\psi_j(x)|^2 dx\,.
% \eeq
Therefore, recalling that $\rho(\Psi,\cdot)$ is non-decreasing and 
using \emph{\ref{i:rho1})}, one gets
\begin{align}
|\rho(\Psi,t)-\rho(\Psi,s)| =&
\|\Psi\|_{B(\y(\Psi,t),t)}^2-\|\Psi\|_{B(\y(\Psi,s),s)}^2
\nonumber
\\
=&  \|\Psi\|_{B(\y(\Psi,t),s)}^2+ \|\Psi\|_{B(\y(\Psi,t),t)\backslash
B(\y(\Psi,t),s)}^2-\|\Psi\|_{B(\y(\Psi,s),s)}^2
\nonumber
\\
\label{e:hyper}
\leq& \|\Psi\|_{B(\y(\Psi,t),t)\backslash B(\y(\Psi,t),s)}^2\,,
\end{align}
where we used the fact that $
\|\Psi\|_{B(\y(\Psi,t),s)}^2-\|\Psi\|_{B(\y(\Psi,s),s)}^2 \leq 0$, by
the definition of $\y(\Psi,s)$.
%  because, by
% the definition of $\y(\Psi,s)$,  $\|\Psi\|_{B(\y(\Psi,s),s)}^2 \geq 
% \|\Psi\|_{B(\underline z,s)}^2$ for any $\underline z\in\GG$. 
For $2<p<\infty$,
by 
%equation \eqref{e:B-B} and 
H\"older's inequality 
%(with $q=\frac{p}{2}$ and $\frac{1}{r}+\frac{1}{q}=1$) 
one has that for any $\y\in\GG$
\begin{align*}
\|\Psi\|_{B(\y,t)\backslash B(\y,s)}^2 \leq&
% \bigg(\int_{\{x\in\RE_+\textrm{ s.t. } s\leq |x-y|<t\}} 1 dx
% \bigg)^{\frac{p-2}{p}}  \|\psi_k\|_{L^{p}(\RE_+)}^2 + \sum_{j\neq k, j=1}^N
% \bigg(\int_{\{x\in\RE_+\textrm{ s.t. } s\leq x+y <t\}}1 dx\bigg)^{\frac{p-2}{p}}
% \|\psi_j\|_{L^p(\RE_+)}^2 \\
% \leq & \bigg(\int_{\{x\in\RE\textrm{ s.t. } s\leq |x-y|<t\}} 1 dx
% \bigg)^{\frac{p-2}{p}}  \|\psi_k\|_{L^{p}(\RE_+)}^2 + \sum_{j\neq k, j=1}^N
% \bigg(\int_{\{x\in\RE\textrm{ s.t. } s\leq x+y <t\}}1 dx\bigg)^{\frac{p-2}{p}}
% \|\psi_j\|_{L^p(\RE_+)}^2\\
%=&
|t-s|^{\frac{p-2}{p}}\sum_{ j=1}^N  \|\psi_j\|_{L^p(\RE_+)}^2 \leq 
N^{\frac{p-2}{p}} |t-s|^{\frac{p-2}{p}} \|\Psi\|_{p}^2.
\end{align*}
% where in the latter inequality we used  that, again by H\"older inequality,
% given $N$ positive numbers $a_j$, $\sum_{j=1}^N a_j^2 \leq N^{\frac{p-2}{p}}
% (\sum_{j=1}^N a_j^p)^{\frac{2}{p}}$ and the definition of
% $\|\cdot\|_{L^p(\GG)}$. 
The latter inequality together with  \eqref{e:hyper}
gives \eqref{e:ii} for $2<p<\infty$. For $p=2$, \eqref{e:ii} is a trivial
consequence of $\rho(\Psi,t) \leq M[\Psi]$. Finally, by \eqref{e:hyper}  and
$\|\Psi\|_{B(\y,t)\backslash B(\y,s)}^2 \leq N |t-s| \|\Psi\|_{\infty}$ we get
\eqref{e:ii-infty}, which concludes the proof of  \emph{\ref{i:rho2})}.
\end{proof}

For any sequence $\Psi_n \in L^2$ we define the concentrated mass
parameter $\tau$ as  
\beq
\label{e:end-3}
\tau = \lim_{t\to\infty} \liminf_{n\to\infty} \rho(\Psi_n,t)\,.
\eeq
As pointed out in the introduction $\tau$ plays a key role in
concentration-compactness lemma because it distinguishes the
occurrence of vanishing, dichotomy or compactness in $H^1$-bounded
sequences. The following lemma, that  replicates
Lem. 1.7.5 in  \cite{Caz03} on a star graph, shows that $\tau$ can be computed as
the limit of $\rho$ on a suitable subsequence.

%%%%%
%LEMMA
%%%%%
\begin{lemma}
\label{l:175}
Let $m>0$ and  $\{\Psi_n\}_{n\in\NA}$ be such that: $\Psi_n\in H^1$,
\beq
\label{e:end-1}
M[\Psi_n] \to m\,,
\eeq
and
\beq
\label{e:end-2}
\sup_{n\in\NA}\|\Psi_n'\|<\infty\,.
\eeq
Then there exist a subsequence $\{\Psi_{n_k}\}_{k\in\NA}$, a nondecreasing
function $\gamma(t)$, and a sequence $t_k\to\infty$ with the following
properties:
\begin{enumerate}[i)]
\item
\label{i:grass-1}
 $\rho(\Psi_{n_k},\cdot)\to \gamma(\cdot)\in [0,a]$ as $k\to\infty$
  uniformly on bounded sets of $[0,\infty)$.
\item 
\label{i:grass-2}
$\tau =
\lim_{t\to\infty}\gamma(t)=\lim_{k\to\infty}\rho(\Psi_{n_k},t_k)=\lim_{
k\to\infty}\rho(\Psi_{n_k},t_k/2)$.
\end{enumerate}
\end{lemma}
%%%%%%%
%PROOF
%%%%%%%
\begin{proof}
From \eqref{e:end-3} there exist $t_k\to\infty$ such that 
\beq
\label{e:1}
\tau= \lim_{k\to\infty} \rho(\Psi_{n_k},t_k) \,.
\eeq
By $0<\rho(\Psi_{n_k},\cdot)\leq \|\Psi_{n_k}\|^2$ the sequence
$\{\rho(\Psi_{n_k},\cdot)\}_{k\in\NA}$  is uniformly
bounded. Moreover, by Gagliardo-Nirenberg inequality there exists a
constant $c$ such that 
$\|\Psi\|_{\infty} \leq c
\|\Psi'\|^{\frac12}\|\Psi\|^{\frac12}$. From this fact, from assumptions
\eqref{e:end-1} and \eqref{e:end-2}, and from \eqref{e:ii-infty}, it
follows that 
the sequence $\rho(\Psi_{n_k},\cdot)$ is also  equicontinuous. Then by
Arzel\`a - Ascoli theorem there exists a subsequence, which we denote by
$\{\rho(\Psi_{n_k},\cdot)\}_{k\in\NA}$ again, which satisfies
\eqref{e:1} and  converges uniformly to some function $\gamma(\cdot)$
on any bounded 
subset of $[0,\infty)$. Since $\rho(\Psi_{n_k},\cdot)$ is nondecreasing, so is
$\gamma(\cdot)$, and the proof of \emph{\ref{i:grass-1})} is complete.

 To prove  \emph{\ref{i:grass-2})},  first we notice that for any
 fixed $t$
 and $k$ large enough such that  $t_k>t$, and since
 $\rho(\Psi_{n_k},\cdot)$ is non 
decreasing, one has 
\[
\rho(\Psi_{n_k},t)  \leq  \rho(\Psi_{n_k},t_k) \,.
\]  
Taking first the limit  $k\to\infty$, then the limit $t\to\infty$, one has 
\[
\lim_{t\to\infty}\gamma(t) \leq \lim_{k\to\infty}  \rho(\Psi_{n_k},t_k) =\tau\,,
\] 
where we used $\lim_{k\to\infty}\rho(\Psi_{n_k},t)=\gamma(t) $ and
\eqref{e:1}. On the other hand, for every $t>0$,
\[
\liminf_{k\to\infty}\rho(\Psi_{n_k},t)\geq
\liminf_{n\to\infty}\rho(\Psi_{n},t)\,,
\]
and taking the limit for $t\to\infty $
\[
\lim_{t\to\infty}\gamma(t)\geq \tau\,,
\]
where we used again  $\lim_{k\to\infty}\rho(\Psi_{n_k},t)=\gamma(t) $ and
\eqref{e:end-3}. Then $\lim_{t\to\infty}\gamma(t)=\tau$. Moreover, since
$\rho(\Psi_n,\cdot)$ is nondecreasing, 
\[
\limsup_{k\to\infty}  \rho(\Psi_{n_k},t_k/2) \leq \limsup_{k\to\infty} 
\rho(\Psi_{n_k},t_k) =\tau\,.
\] 
On the other hand, for fixed $t>0$ and $k$ large enough one has $t_k/2>t$ and 
\[
 \rho(\Psi_{n_k},t_k/2) \geq  \rho(\Psi_{n_k},t)
\]
taking first the $\liminf_{k\to\infty}$ then the limit for $t\to\infty$ it
follows that 
\[
\liminf_{k\to\infty} \rho(\Psi_{n_k},t_k/2) \geq \lim_{t\to\infty} \gamma(t)
=\tau
\]
then $\lim_{k\to\infty} \rho(\Psi_{n_k},t_k/2)=\tau$, which concludes the proof
of \emph{\ref{i:grass-2})}.
\end{proof}

We are now ready to prove the concentration-compactness lemma.
%%%%%
%LEMMA
%%%%%
\begin{lemma}[Concentration-compactness]
\label{l:cc}
Let $m>0$ and $\{\Psi_n\}_{n\in\NA}$ be such
that: $\Psi_n\in\mathcal E$,
\[
M[\Psi_n] \to m\,,
\]
\[
\sup_{n\in\NA}\|\Psi_n'\|<\infty\,.
\]
Then there exists a subsequence $\{\Psi_{n_k}\}$ such that:
%{\bf ***}
\begin{enumerate}[i)]
\item\label{i:cc1}
(Compactness) If $\tau=m$, at least one of
    the two following cases occurs:
\begin{itemize}
\item[$i_1)$](Convergence) There exists a
function
$\Psi\in \mathcal E$  such that $\Psi_{n_k}\to \Psi$  in $L^p$ as
$k\to\infty$ for all $2\leq p\leq \infty$ .
%\footnote{Sappiamo se c'e' anche convergenza forte in H1? Soprassediamo.}
\item[$i_2)$](Runaway) There exists $j^*$, such that for any $j\neq j^*$ and
$2\leq p\leq\infty$
\beq
\label{e:run-1}
\|\psi_{n_k,j}\|_{L^p(\RE_+)} \to 0\,,
\eeq
moreover for any $t>0$
\beq
\label{e:run-2}
\|\Psi_{n_k}\|_{L^p(B(\underline 0,t))} \to 0 \,.
\eeq
\end{itemize}
\item\label{i:cc2} (Vanishing) If $\tau=0$, then  $\Psi_{n_k}\to 0 $ in $L^p$ as
$k\to\infty$  for all $2< p\leq \infty$. 
\item\label{i:cc3} (Dichotomy) If $0<\tau<m$, then there exist two sequences 
$\{\VV_k\}_{k\in\NA}$ and $\{\WW_k\}_{k\in\NA}$ in $\mathcal E$
such that 
\beq
 \supp \VV_k \cap \supp \WW_k = \emptyset 
\label{dic1} 
\eeq
\beq
|\VV_k(x,j)| + |\WW_k(x,j)| \leq |\Psi_{n_k}(x,j)| \qquad
\textrm{for any } j=1,...,N\,;\; x\in\RE_+ 
\label{dic2} 
\eeq
\beq
\| \VV_k \|_{H^1} + \|\WW_k\|_{H^1} \leq c
\|\Psi_{n_k}\|_{H^1}
\label{dic3} 
\eeq
\beq
\lim_{k \to \infty} M [\VV_k ] = \tau \qquad \qquad \lim_{k \to \infty} M[
\WW_k ]= m -\tau
\label{dic4} 
\eeq
\beq
\liminf_{k\to \infty} \lf( \|\Psi_{n_k}'\|^2 - \| \VV_k' \|^2 - \|
\WW_k' \|^2 \ri) \geq 0
\label{dic5}
\eeq
\beq
\lim_{k\to \infty} \lf( \|\Psi_{n_k} \|_{p}^p - \| \VV_k \|_{p}^p - \|
\WW_k \|_{p}^p \ri) =0 \qquad 2 \leq p < \infty
\label{dic6}
\eeq
\beq
\label{dic7}
\lim_{k\to\infty}\left||\Psi_{n_k}(0,j)|^2-  |\VV_{k}(0,j)|^2 -| \WW_{k}(0,j)|^2\right|=0
\qquad
\textrm{for any } j=1,...,N\,.
\eeq
\end{enumerate}
\end{lemma}
%%%%%
%PROOF
%%%%%
\begin{proof}
%[Proof of Lemma \ref{l:cc}]
Let $\{\Psi_{n_k}\}_{k\in\NA}$, $\gamma(\cdot)$ and $t_k$  be the subsequence, the
function and the sequence defined in Lem. \ref{l:175}.  

Proof of \emph{\ref{i:cc1})}. Suppose $\tau=m$. By Lem. \ref{l:175}
\emph{\ref{i:grass-2})},  for any  $m/2\leq\lambda< m$ there exists $t_\la$
large enough such that $\gamma(t_\la)>\la$. Then by Lem.  \ref{l:175}
\emph{\ref{i:grass-1})}, for $k$ large enough $\rho(\Psi_{n_k},t_\la)>\la$. 

Set $\underline{y_k}(t)\equiv \underline{y}(\Psi_{n_k},t)$, where
$\underline{y}(\Psi_{n_k},t)$ was defined in Prop. \ref{p:rho}
\emph{\ref{i:rho1})}. For $k$ large enough, we have that 
\beq
\label{e:cmin}
d(\underline{y_{k}}(t_{m/2}),\underline{y_{k}}(t_{\la})) \leq t_{m/2}+t_\la\,. 
\eeq
To prove \eqref{e:cmin}, assume 
$d(\underline{y_{k}}(t_{m/2}),\underline{y_{k}}(t_{\la})) > t_{m/2}+t_\la$,
then the balls $B(\underline{y_k}(t_{m/2}),t_{m/2})$ and
$B(\underline{y_k}(t_\la),t_{\la})$ would be disjoint, thus implying  
\[
M[\Psi_{n_k}]\geq \|\Psi_{n_k}\|_{B(\underline{y_k}(t_{m/2}),t_{m/2})}^2 +
\|\Psi_{n_k}\|_{B(\underline{y_k}(t_\la),t_{\la})}^2 > \frac{m}{2}+\la \geq m
\]
which is impossible because $M[\Psi_{n_k}]\to m$.
Next we distinguish two cases: $\underline{y_{k}}(t_{m/2})$ bounded
and $\underline{y_k}(t_{m/2})$ unbounded.

Case $\underline{y_{k}}(t_{m/2})$ bounded. We first recall that
$\Psi_{n_k}(\cdot,j)\in H^1(\RE_+)$, then by \cite[Th. VIII.5]{Bre83}
we can extend each $\Psi_{n_k}(\cdot,j)$ to  an even  function
$\widetilde\Psi_{n_k}(\cdot,j)\in H^1(\RE)$,  in such a way that
the sequence $\widetilde\Psi_{n_k}(\cdot,j)$ is uniformly bounded in $
H^1(\RE)$.  Applying
\cite[Cor. 5.5.2 and Lem. 5.5.3, see also Th. 5.1.8]{Caz06} to each sequence  $\{\widetilde\Psi_{n_k}(\cdot,j)\}_{k\in\NA}$ we get that there exist $\widetilde\Psi(\cdot,j)\in H^1(\RE)$ such that, up to taking a subsequence,  $\widetilde\Psi_{n_k}(\cdot,j)\to \widetilde \Psi(\cdot,j)$ in $L^2([-A,A])$ for any $A>0$. Restricting each $\widetilde\Psi_{n_k}(\cdot,j)$ and  $\widetilde \Psi(\cdot,j)$  to $\RE_+$ we get that there exists $\Psi\in
H^1$ and a subsequence, which we still
denote by $\{\Psi_{n_k}\}_{k\in\NA}$,
such
that  $\Psi_{n_k}\to \Psi$ in $L^2(B(\underline y,t))$, for any fixed
$\underline y$ and $t$. Moreover, again by \cite[Lem. 5.5.3]{Caz06}, we have that $\widetilde\Psi_{n_k}(\cdot,j)$ converges to $\widetilde \Psi(\cdot,j)$ weakly in $H^1(\RE)$. Then by the Rellich-Kondrashov theorem \cite[Th. 8.9]{LL01}, $\widetilde\Psi_{n_k}(0,j)$ converges to $\widetilde \Psi(0,j)$. Since $\Psi_{n_k}\in\EE$ one has  $\widetilde\Psi_{n_k}(0,j)=\widetilde\Psi_{n_k}(0,j')$, then the same is true also for $\widetilde \Psi(0,j)$, thus implying  $\Psi\in\EE$. The function $\Psi$ might be the null function, next we show that for $\underline{y_{k}}$ bounded this is not the case. We prove indeed that $M[\Psi]=m$ and therefore $\Psi_{n_k} \to
\Psi$ in $L^2$. 
%Since by \eqref{e:cmin} for any $t\geq t_{m/2}$, $\underline{y_{k}}(t)$ is bounded, up to choosing a subsequence which we still denote by $\Psi_{n_k}$, we
%can
%assume that $\underline{y_{k}}(t)\to \underline{y}^*(t)$.
Fix $\lambda \in (m/2,m)$, and let
  $t_\lambda$ be such that $\rho (\Psi_{n_k}, t_\lambda) > \lambda$
  eventually in $k$. Since, by \eqref{e:cmin},
  $\underline{y_{k}} (t_\lambda)$ is bounded,  up to choosing a
subsequence which we  still denote by $\Psi_{n_k}$, we
can
assume that $\underline{y_{k}}(t_\lambda)\to
\underline{y}^*(t_\lambda)$ and  $\underline{y_{k}}(t_{m/2})\to
\underline{y}^*(t_{m/2})$.  Then, fixed $\ve>0$,  for $k$ large enough
we
have $d(\underline{y}^*(t_{m/2}),\underline{y_k}(t_{m/2}))\leq
\ve$, so that, by
\eqref{e:cmin} and the triangle inequality, $d(\underline{y}^*(t_{m/2}),
\underline{y_k}(t_\la)) \leq \ve + t_{m/2} + t_\la$. Setting $T= 2(\ve +
t_{m/2} + t_\la) $ we certainly have that
$B(\underline{y_k}(t_\la),t_\la)\subseteq B(\underline{y}^*(t_{m/2}),T)$ so that
\beq
\label{e:bren}
\|\Psi_{n_k}\|_{B(\underline{y}^*(t_{m/2}),T)}^2 \geq  
\|\Psi_{n_k}\|_{B(\underline{y_k}(t_{\la}),t_\la)}^2 =
\rho(\Psi_{n_k},t_\la)>\la\,. 
\eeq
 Then by inequality \eqref{e:bren} and 
 since 
\[
M[\Psi] \geq \|\Psi\|_{B(\underline{y}^*(t_{m/2}),T)}^2 =\lim_{k\to\infty}
\|\Psi_{n_k}\|_{B(\underline{y}^*(t_{m/2}),T)}^2
\]
 we have that  $M[\Psi] \geq \la$. As we can
   choose $\lambda$ arbitrarily close to $m$, we get $M[\Psi] \geq
m$. On the other hand, by weak convergence,  we have that 
\[
M[\Psi] \leq \liminf_{k\to\infty}M[\Psi_{n_k}]=m\,.
\]
So that $M[\Psi] =m$ and by \cite[Lem. 5.5.3]{Caz06} we get $\Psi_{n_k} \to
\Psi$ in $L^2$. The convergence in $L^p$ for $2<p\leq \infty$ follows from
Gagliardo-Nirenberg inequality. 

Assume now that  $\underline{y_{k}}(t_{m/2})$ is
unbounded. We shall adapt the argument used in
  the case of $\underline{y_{k}}(t_{m/2})$
  bounded. Denote
$\underline{y_k}(t_{m/2})=(y_k(t_{m/2}),j_k(t_{m/2}))$.  Up to choosing a subsequence which we still
denote by $\Psi_{n_k}$, we
can
assume that there exists $j^*$ such that  $j_k(t_{m/2})=j^*$ and
$y_k(t_{m/2})\to\infty$.
%For any $t\geq t_{m/2}$, $T_{max}>0$ and $k$ large enough we
%have
%then $y_k(t) > T_{max}$.
Take  $m/2<\la<m$, set  $T_{max}>4
\max\{t_\la,t_{m/2}\}$  and notice that, due to
  \eqref{e:cmin}, the sequence $\underline{y_{k}}(t_{\lambda})$ diverges 
  on the ${j^*}$-th edge.
Define $\widetilde \psi_{n_k}\in L^2(\RE_+)$ by
\[
\widetilde \psi_{k}(x) = \psi_{j^*,n_k}(x+y_k(t_{m/2})-T_{max})\,.
\]
We notice that for 
%any $m/2\leq \la < m$ and
$k$ large enough
\beq
\label{e:ever}
\rho(\Psi_{n_k},t_\la) = \|\Psi_{n_k}\|_{B(\underline{y_k}(t_\la),t_\la)}^2 = 
 \|\psi_{j^*,n_k}\|^2_{L^2((y_k(t_\la)-t_\la,y_k(t_\la)+t_\la))} \,,
\eeq
then by an argument similar to the one used above  we
have that, for $T=2(t_{m/2}+t_\la)$ and using the fact that $T_{max} > T$,
\[
\|\widetilde \psi_k\|^2_{L^2((T_{max}-T,T_{max}+T))} \geq
\|\psi_{j^*,n_k}\|^2_{L^2((y_k(t_\la)-t_\la,y_k(t_\la)+t_\la))} > \la
\]
where in the latter inequality we used equation \eqref{e:ever}. Applying
\cite[Cor. 5.5.2 and Lem. 5.5.3]{Caz06} to $\RE_+$,  we get that there exists
$\psi\in
H^1(\RE_+)$ and a subsequence, which we still denote by
$\{\widetilde \psi_{k}\}_{k\in\NA}$,
such
that  $\widetilde \psi_{k}\to \psi$ in $L^2((T_{max}-T,T_{max}+T))$, for any
fixed
$T_{max}>T$. Then, following what was done in the case $\underline{y_k}$ bounded,
we prove
that $ \|\psi\|^2_{L^2(\RE_+)} =m$ and by \cite[Lem. 5.5.3]{Caz06} we get
$\widetilde \psi_{k} \to
\psi$ in $L^2(\RE_+)$. Also in this case the convergence $\widetilde \psi_{k}
\to
\psi$  in $L^p(\RE_+)$ for
$2<p\leq \infty$ follows from
Gagliardo-Nirenberg inequalities.

\n
To get \eqref{e:run-1} and
\eqref{e:run-2} for $p=2$ we notice that for any $\ve>0$ and $k$ large
enough $M[\Psi_{n_k}] < m+\ve$. Set $\la=m-\ve$. From the discussion above in
the unbounded case we deduce that for any $t$ and  $k$ large enough
$y_k(t_{m/2}) - T_{max} > t$, moreover
\[
 \int_{t}^\infty |\psi_{n_k,j^*}(x)|^2 dx \geq 
\int_{y_k(t_{m/2})-T_{max}}^\infty |\psi_{n_k,j^*}(x)|^2 dx =  
 \|\widetilde \psi_k\|_{L^2(\RE_+)} > \la=m-\ve\,.
\] 
Then, by 
\[
M[\Psi_{n_k}] = \sum_{j\neq j^*} \|\psi_{n_k,j}\|^2_{L^2(\RE_+)} +
 \int_{0}^t |\psi_{n_k,j^*}(x)|^2 dx + \int_{t}^\infty |\psi_{n_k,j^*}(x)|^2
dx
< m+\ve 
\]
we get 
\[
\sum_{j\neq j^*} \|\psi_{n_k,j}\|^2_{L^2(\RE_+)} +
 \int_{0}^t |\psi_{n_k,j^*}(x)|^2 dx
< 2\ve \,.
\]
The limits \eqref{e:run-1} and
\eqref{e:run-2} for  $p>2$
follow by  Gagliardo-Nirenberg inequalities.  
\\

Proof of \emph{\ref{i:cc2})}. Suppose $\tau=0$.  By Lem. \ref{l:175},
$\tau=\lim_{k\to\infty}\rho(\Psi_{n_k},t_k)=0$.  Then since $\rho(\Psi,\cdot)$
is non-decreasing, $\lim_{k\to\infty}\rho(\Psi_{n_k},1)=0$. By H\"older
inequality: for $2< r< 6$, $\|\Psi\|_{r} \leq
\|\Psi\|_{6}^{\frac{3(r-2)}{2r}}\|\Psi\|^{\frac{(6-r)}{2r}}$ and, for
$6< r< \infty$, $\|\Psi\|_{r} \leq
\|\Psi\|_{6}^{\frac{6}{r}}\|\Psi\|_{\infty}^{\frac{r-6}{r}}$.
We claim that 
\beq
\label{e:no}
\|\Psi\|_{6}^6 \leq c \rho(\Psi,1)^2\left(\|\Psi'\|^2+\|\Psi\|^2\right)
\eeq
Then \emph{\ref{i:cc2})} follows by recalling that by \cite[Th. VIII.7]{Bre83}  one
has $\|\Psi\|_{\infty}\leq c  (\|\Psi\|+\|\Psi'\|)$. Moreover by Gagliardo-Nirenberg inequalities one has $\|\Psi\|_{\infty} \leq c \|\Psi'\|^a\|\Psi\|_{p}^{1-a}$ for some $p>2$, then $\Psi_{n_k}\to 0 $ in $L^\infty$ as well.

\n Next we prove inequality \eqref{e:no}. Let $\{I_k\}_{k=1}^\infty$ be a sequence of unit intervals such that $I_j\cap
I_k=\emptyset$ for all $j\neq k$ and $\overline{\cup_{k=1}^\infty I_k} = \RE_+$.
Moreover, let us denote by $\|\Psi\|_{L^p(\GG_k)}$ the norm defined by
$\|\Psi\|_{L^p(\GG_k)}^p = \sum_{j=1}^N \|\psi_j\|_{L^p(I_k)}^p$, $1\leq p \leq
\infty$. 
By \cite[Th. VIII.7]{Bre83} there exists a constant $c$ such that
$\|\psi_j\|_{L^\infty(I_k)} \leq c 
(\|\psi_j\|_{L^1(I_k)}+\|\psi_j'\|_{L^1(I_k)})$, where $c$ is
independent of
$k$. Then $\|\Psi\|_{L^\infty(\GG_k)} = \sup_{1\leq j\leq N }
\|\psi_j\|_{L^\infty(I_k)} \leq c 
\sum_{j=1}^N(\|\psi_j\|_{L^1(I_k)}+\|\psi_j'\|_{L^1(I_k)}) = c 
(\|\Psi\|_{L^1(\GG_k)}+\|\Psi'\|_{L^1(\GG_k)})$, where $c$ is the same constant
as above. Changing $\Psi$ into $|\Psi|^2\equiv(|\psi_1|^2,...,|\psi_N|^2)$ in
the latter inequality and using $|\nabla |\psi_j|^2|\leq 2|\psi_j| |\nabla
\psi|$ and H\"older's inequality we get  $\|\Psi\|_{L^\infty(\GG_k)}^2 \leq   c
 \sum_{j=1}^N\|\psi_j\|_{L^2(I_k)}(\|\psi_j\|_{L^2(I_k)}+\|\psi_j'\|_{L^2(I_k)})
\leq    c  \|\Psi\|_{L^2(\GG_k)}
(\|\Psi\|_{L^2(\GG_k)}+\|\Psi'\|_{L^2(\GG_k)})$. Finally, from the latter
inequality, using H\"older's inequality $\|\Psi\|_{L^6(\GG_k)}^6\leq
\|\Psi\|_{L^2(\GG_k)}^2\|\Psi\|_{L^\infty(\GG_k)}^4$, we end up with 
\[
\|\Psi\|_{L^6(\GG_k)}^6\leq  c  \|\Psi\|_{L^2(\GG_k)}^4 
(\|\Psi\|_{L^2(\GG_k)}^2+\|\Psi'\|_{L^2(\GG_k)}^2)
\]
where the constant $c$ is the same as above and is independent of $k$. Summing
on $k$ we get 
\[
%\label{e:swamp}
\|\Psi\|_{6}^6\leq c  \sup_{k}\|\Psi\|_{L^2(\GG_k)}^4 
(\|\Psi\|^2+\|\Psi'\|^2)
\leq  c N^2  \rho(\Psi,1)^2 (\|\Psi\|^2+\|\Psi'\|^2).
\]

Proof of \emph{\ref{i:cc3})}. Let $\theta$ and $\varphi$ be two cut-off functions such that $\theta,\varphi\in C^{\infty} (\RE )$, 
$0\leqslant \theta , \varphi \leqslant 1$ and 
\[
\theta(t) = 
\begin{cases}
1 & 0\leqslant |t| \leqslant 1/2 \\
0 & |t| \geqslant 3/4 
\end{cases}
\qquad \qquad 
\varphi(t) = 
\begin{cases}
0 & 0\leqslant |t| \leqslant 3/4 \\
1 & |t| \geqslant 1 
\end{cases} 
\]
Take $ t_k$ as in equation \eqref{e:1} and set $\underline{y}(t_k)\equiv \underline{y}(\Psi_{n_k},t_k)$, where
$\underline{y}(\Psi_{n_k},t)$ was defined in Prop. \ref{p:rho}
\emph{\ref{i:rho1})}. We shall write $\underline{y}(t_k)=(y(t_k),j(t_k))$. Define  the following cut off functions
\[
\theta_k (x) = \theta \lf( \frac{x - y( t_k /2)}{t_k} \ri) \qquad \qquad 
\varphi_k (x) = \varphi \lf( \frac{x - y( t_k /2)}{t_k} \ri) 
\]
and 
\[
\tilde \theta_k (x) = \theta \lf( \frac{x + y( t_k /2)}{t_k} \ri) \qquad \qquad 
\tilde \varphi_k (x) = \varphi \lf( \frac{x + y( t_k /2)}{t_k} \ri) .
\]
Let $\VV_k=(\VV_k(\cdot,1),...,\VV_k(\cdot,N))$ be defined by 
\[
\begin{aligned}
&\VV_{k}(x,j(t_k/2)) = \theta_k(x)\Psi_{n_k}(x,j(t_k/2))\\
&\VV_{k}(x,l) = \tilde \theta_k(x)\Psi_{n_k}(x,l) \qquad\textrm{for any} \quad l\neq j(t_k/2).
\end{aligned}
\]
Moreover, let  $\WW_k=(\WW_k(\cdot,1),...,\WW_k(\cdot,N))$  be  defined by
\[
\begin{aligned}
&\WW_{k}(x,j(t_k/2)) = \varphi_k(x)\Psi_{n_k}(x,j(t_k/2))\\
&\WW_{k}(x,l) = \tilde\varphi_k(x)\Psi_{n_k}(x,l) \qquad\textrm{for
    any} \quad l\neq j(t_k/2). 
\end{aligned}
\]
We remark that $\VV_k$ ($\WW_k$ resp.) coincides with $\Psi_{n_k}$ in
the ball $B(\underline y(t_k/2),t_k/2)$ (in the set $\GG\backslash
B(\underline y(t_k/2),t_k)$ resp.) and $\VV_k = 0$ ($\WW_k = 0$ resp.)
in  the set $\GG\backslash B(\underline y(t_k/2),3t_k/4)$ (in the ball
$B(\underline y(t_k/2),3t_k/4)$ resp.). Properties \eqref{dic1},
\eqref{dic2} and \eqref{dic3} are immediate. Next we notice that by
Prop. \ref{p:rho}, \emph{\ref{i:rho1})}, 
\[
\rho ( \Psi_{n_k} , t_k /2 )  = \|\Psi_{n_k}\|_{B(\underline y(t_k/2),t_k/2)}^2 
\leqslant M[V_k]\,. 
\]
Moreover, since  $\theta(t) \leq 1$,  
\[
M[V_k] \leq  \|\Psi_{n_k}\|_{B(\underline y(t_k/2),t_k)}^2 \leq \|\Psi_{n_k}\|_{B(\underline y(t_k),t_k)}^2 = \rho( \Psi_{n_k} , t_k)\,,
\]
where in the latter inequality we have taken into account 
 the optimality of $y(t_k)$ according to Prop. \ref{p:rho}, \emph{\ref{i:rho1})} and to the definition of $\rho(\Psi,t)$. Therefore 
\be
\lim_{k\to \infty}M[V_k] = \tau
\ee
by Lem. \ref{l:175}, \emph{\ref{i:grass-2})}. Now put $Z_k \equiv \Psi_{n_k} - V_k - W_k $ and notice that $\supp Z_k \subseteq B(\underline y(t_k/2),t_k)\backslash B(\underline y(t_k/2),t_k/2)$ and  $|Z_k|\leqslant |\Psi_{n_k}|$, to be understood pointwise.
Then one has 
\begin{align}
M[Z_k] \leq &  \|\Psi_{n_k}\|^2_{B(\underline y(t_k/2),t_k)\backslash B(\underline y(t_k/2),t_k/2)}
\nonumber
\\
= & \|\Psi_{n_k}\|^2_{B(\underline y(t_k/2),t_k)} - 
     \|\Psi_{n_k}\|^2_{B(\underline y(t_k/2),t_k/2)}
\leq \rho(\Psi_{n_k} , t_k ) - \rho(\Psi_{n_k} , t_k /2) 
\label{zampone}
\end{align}
again by the optimality properties of $\underline y(t_k)$. It follows from \eqref{zampone} and Lem. \ref{l:175}, \emph{\ref{i:grass-2})}  that $M[Z_k]\to 0 $, 
and therefore $M[W_k]\to m -\tau$ which concludes the proof of \eqref{dic4}.
Equation \eqref{dic6} follows by  
\[
\lf|
| \Psi_{n_k} |^p - |V_k|^p -|W_k|^p
\ri| \leqslant c |\Psi_{n_k}|^{p-1} |Z_k|\,,
\]
to be understood pointwise, and  H\"older inequality. Moreover,  since
$\|Z_k\|_{H^1} \leq c$ and $\|Z_k\|\to 0$ then 
$\|Z_k\|_{L^\infty} \to 0$ by Gagliardo-Nirenberg inequality.
Therefore
\[
%\label{nodico3}
|\Psi_{n_k}(0,j)|^2\equiv |\ZZ_{k}(0,j)-\VV_{k}(0,j)-\WW_{k}(0,j)|^2\to
|\VV_{k}(0,j)+\WW_{k}(0,j)|^2 =  |\VV_{k}(0,j)|^2 +| \WW_{k}(0,j)|^2\,,
\]
from which \eqref{dic7}.

Concerning the  inequality \eqref{dic5},
first notice that% (here we use the notation $\Psi(j)' = \frac{d}{dx} \Psi(x,j)$)
\[
\begin{aligned}
&|\Psi_{n_k}'(\cdot,j)|^2 - |\VV_k'(\cdot,j)|^2 - |\WW_k'(\cdot,j)|^2 \\
=  &|\Psi_{n_k}'(\cdot,j)|^2 (1-\theta_k^2 -\varphi_k^2) 
-|\Psi_{n_k}(\cdot,j)|^2 (|\theta_k '|^2 + |\varphi_k '|^2) 
- \Re(\overline\Psi_{n_k}(\cdot,j) \Psi_{n_k}'(\cdot,j)\,(\theta_k^2 +\varphi_k^2 )' \\
\geqslant &-\frac{c}{t_k^2} |\Psi_{n_k}(\cdot,j)|^2
- \frac{c}{t_k} |\Psi_{n_k}'(\cdot,j) | |\Psi_{n_k}(\cdot,j)|\,.
\end{aligned}
\]
Summing up on $j$ we obtain \eqref{dic5}.
\end{proof}

%%%%%%
%SECTION
%%%%%%
%\section{Proof of Theorem \ref{t:prob1}} \label{sec4}
\section{Constrained energy minimization} \label{sec4}

In this section we prove that for a small enough mass there  exists a
solution to the constrained energy minimization problem. The proof is
inspired by the work of Cazenave-Lions for the NLS in $\RE$, see in
particular Prop. 8.3.6 in \cite{Caz03}. Nevertheless, due to the lack
of translational invariance and to the presence of a singular
potential well in the vertex, several non trivial changes will be
necessary. Some adjustments were already implemented in the
concentration-compactness lemma, to resolve the ambiguity of the case
$\tau=m$. To prove Th. \ref{t:prob1}, another major adjustment will be
necessary, i.e. we have to prove that runaway subsequences are not minimizing if the mass is small enough.
To prove the existence of a minimizer of $E$,  we use the concentration-compactness result as follows. We
assume that $\{\Psi_n\}_{n\in \NA}$ is such that $M[\Psi_n] \to m$,
$\|\Psi_n\|_{H^1}$ is bounded and $\{\Psi_n\}_{n\in \NA}$ is a minimizing
sequence for the energy functional, thus any subsequence of
$\{\Psi_n\}_{n\in\NA}$ is a minimizing sequence as well. By using the energy
functional we prove that the concentrated mass parameter $\tau$ of a minimizing
sequence must equal $m$, so that for minimizing sequences the vanishing and
dichotomy cases cannot occur. Then, if $\{\Psi_n\}_{n\in\NA}$ is a minimizing sequence,
we are in the compactness case. 
%At this stage the concentration-compactness
%lemma still leaves several possibilities to the subsequences of the minimizing
%sequence $\{\Psi_n\}_{n\in\NA}$. More precisely, the lemma implies that there
%exists at least one subsequence which is convergent or runaway, without
%excluding the possibility that  subsequences with different behavior
%exist. Again by using the properties of the  energy functional we will show that
%this possibility is excluded. This is done by proving 
In order to distinguish between the two subcases of convergence and
runaway, we prove
that there exists a
critical value of the mass $m^*$ such that  if $m<m^*$ then the infimum of the
energy functional is attained by convergent sequences. 
%while for $m>m^*$ the infimum is attained by runaway sequences. 
%Due to the quantitative estimates needed to get this last part of the result, it comes into play the explicit knowledge of some solitary solutions, or asymptotically approximate solutions of equation \eqref{schrod}, 
The explicit expression of $m^\ast$ comes from the knowledge of the stationary states of equation \eqref{schrod}
obtained in \cite{[ACFN4]}. 
%In particular, it is possible to establish a lower bound for the energy at constant mass of runaway sequences, which allows to get the minimal mass above which  (or equivalently the maximum strength $\alpha$ of the interaction under which) one has runaway behavior. 
If a minimizing sequence is runaway, then we find that there is no
minimum of the energy but only an infimum value, as runaway sequences
weakly converge 
to 0. An example of this behavior for cubic nonlinearity  ($\mu=1$)
and for the case $\alpha=0$ (the so called Kirchhoff or free quantum
graph) was explicitly worked out in \cite{[ACFN2]}. Here it is shown
that the  phenomenon is more general and that a sufficiently deep
potential well at the vertex, i.e. $\al$ negative enough, is needed
in order to prevent a minimizing sequence from escaping to
infinity. We remark that 
apart from the explicit estimate of the bound on the threshold, made possible
by the choice of a delta vertex, the behavior discovered and
studied here appears to be simple and general.

%cost more energy than convergent
                                %subsequences. 

\begin{proof}[Proof of Theorem \ref{t:prob1}]
We prove first that $0<\nu<\infty$. Take $\Psi\in\EE$ such that  $M[\Psi] = m$ and define $\Psi_\la=(\psi_{\la,1},...,\psi_{\la,N})$ with 
$\psi_{\la,j}(x)=\la^{\frac12}\psi_j(\la x)$. Then $\Psi\in\EE$ and $M[\Psi_\la] = m$ as well. It is easy to see that for $0<\mu<2$ and  $\al<0$, one can take $\la$ small enough so that  $E[\Psi_\la]<0$, then $\nu>0$.

\n 
To prove that $\nu<+\infty$ we use first Gagliardo-Nirenberg inequalities which give 
\[
 \|\Psi\|_{{2\mu+2}}^{2\mu+2} \leq c \|\Psi'\|^{\mu} \|\Psi\|^{2+\mu}
\]
and 
\[
%\label{e:horses}
 |\psi_j(0)|^2\leq \|\Psi\|_{\infty}^2\leq c  \|\Psi'\|\|\Psi\|\,.
\]
Then, if $M[\Psi] = m$ we have 
\[
 E[\Psi]
 \geq \frac{1}{2}\| \Psi' \|^2 - \frac{m^{\frac{2+\mu}{2}}}{2\mu+2}c \|\Psi'\|^{\mu}
  - c \sqrt m \frac{|\alpha|}{2} \|\Psi'\|\,.
\]
We notice that for any $a,b,c>0$ and $0< \mu < 2$ there exist $\de,\beta >0$ such that $a x^2 - bx^\mu  -cx > \de x^2 - \beta $, for any $x\geq0$, then  
\beq
\label{e:dec}
 E[\Psi]
 \geq \de \| \Psi' \|^2 -\beta \,,
\eeq
which implies $\nu \leq \beta$.\\

In the  remaining part of the proof we shall prove that for $m<m^*$ minimizing sequences have a  convergent subsequence.

\n
In order to prove Th. \ref{t:prob1} we can consider a slightly more general setting taking
%{\bf ***} 
$\{\Psi_n\}_{n\in\NA}$ be such that $M[\Psi_n] \to m$ and $E[\Psi_n]\to -\nu$. %i.e.,  $\{\Psi_n\}_{n\in\NA}$  is a minimizing sequence.  
We shall prove that exists $\hat \Psi \in H^1 (\GG)$ such that $M[\hat \Psi] = m $, $E[\hat \Psi] =-\nu$ and $\Psi_n \to \hat \Psi$ in $ H^1 (\GG)$.

We can assume that 
$E[\Psi_n] \leq -\nu/2$ then by inequality \eqref{e:dec}, up to taking a subsequence,  we have that $\{\Psi_n\}$ is bounded in $H^1$, moreover the following lower bound holds true
\beq
\label{e:floor}
 \frac{1}{\mu+1} \| \Psi_n\|_{2\mu+2}^{2\mu+2} + |\al| |\psi_{n,1}(0)|^2 \geq
\nu \,.
\eeq

Next we use Lem. \ref{l:cc} and  prove that  vanishing and dichotomy
cannot occur for $\Psi_n$. Set $\tau = \lim_{t\to\infty}\liminf_{n\to\infty}
\rho(\Psi_n,t)$. First we prove that  vanishing cannot occur. If
$\tau=0$,  then by Lem. \ref{l:cc} there would exist a subsequence $\Psi_{n_k}$
such that $\| \Psi_{n_k}\|_{L^p} \to 0 $ for all $2<p\leq \infty$ but
this  would contradict 
\eqref{e:floor}.

To prove that dichotomy cannot occur, suppose $0<\tau<m$, then there
would exist $V_k$ and $W_k$ satisfying \eqref{dic1}-\eqref{dic7}.
In particular we know that
\[
 \liminf_{k\to \infty} \left(\|\Psi_{n_k}' \|^2 - \| V_k' \|^2
- \| W_k' \|^2 \right)  \geq 0
%\label{nodico1}
\]
\[
\lim_{k\to \infty} \lf( \| \Psi_{n_k}\|_p^p - \|V_k \|_p^p -
\| W_k \|_p^p \ri)=0 \qquad 2\leq p < \infty 
%\label{nodico2}
\]
and 
\be
\lim_{k\to \infty}\left||\psi_{1,n_k}(0)|^2-  |v_{1,k}(0)|^2 -| w_{1,k}(0)|^2\right| = 0\,.
\ee
Summing up, we arrive at
\be
\liminf_{k\to\infty} \lf(
E[\Psi_{n_k} ] - E[ V_k] - E[W_k]
\ri) \geq 0 \,,
\ee
which implies
\beq
\label{e:black-1}
\limsup_{k\to\infty} \lf(
 E[ V_k] + E[W_k]
\ri) \leq -\nu \,.
\eeq
Notice that, given $\Psi\in\EE$ with $M[\Psi] =m$ and $\de >0$, then
\[
E[\Psi] = \frac{1}{\de^2} E[\de \Psi] + \frac{\de^{2\mu} -1}{2\mu+2} \| \Psi
\|_{2\mu+2}^{2\mu+2}.
\]
We remark that $V_k, W_k \in H^1$ and they satisfy the right boundary condition at
the vertex since $\Psi_{n_k}$ does and the multiplication with the cut-off functions
preserves that, then $V_k, W_k \in \EE$.
Let $\de_k= \sqrt {m / M[V_k]}$ and $\ga_k = \sqrt {m / M[W_k]}$  such that $M[\de_k V_k] ,\,
M[\ga_k W_k] =m$. Then, using the above equality and the fact that
$E[\de_k V_k], E[\ga_k \WW_k] \geq -\nu$,  one has
\[
E[V_k] \geq - \frac{\nu}{\de^2_k} + \frac{\de^{2\mu}_k -1}{2\mu+2} \| V_k
\|_{2\mu+2}^{2\mu+2}
\]
\[
E[W_k] \geq - \frac{\nu}{\ga^2_k} + \frac{\ga^{2\mu}_k -1}{2\mu+2} \| W_k
\|_{2\mu+2}^{2\mu+2}
\]
from which 
\[
E[V_k]+E[W_k] \geq -\nu \lf( \frac{1}{\de^2_k} + \frac{1}{\ga^2_k} \ri) +
\frac{\de^{2\mu}_k -1}{2\mu+2} \| V_k \|_{2\mu+2}^{2\mu+2} +
\frac{\ga^{2\mu}_k -1}{2\mu+2} \| W_k \|_{2\mu+2}^{2\mu+2}\,.
\]
Notice that  by \eqref{dic4}
\[
\frac{1}{\de^2_k} \to \frac{\tau}{m} \qquad \qquad \frac{1}{\ga^2_k} \to 1-\frac{\tau}{m}\,.
\]
Let $\theta = \min \{ (\tau/m)^{-\mu} , (1-\tau/m)^{-\mu} \}$ and notice that $\theta
>1$ since $0<\tau/m <1$. Therefore
\begin{align}
\label{e:black-2}
\liminf_{k\to\infty} \lf(
 E[ V_k] + E[W_k]
\ri) 
&\geq -\nu + \frac{\theta -1}{2\mu+2} \liminf_{k\to\infty} \| \Psi_{n_k}
\|_{2\mu+2}^{2\mu+2} > -\nu,
\end{align}
where we used the fact that $\liminf_{k\to\infty} \| \Psi_{n_k}
\|_{2\mu+2}^{2\mu+2} \neq 0$. The latter claim is proved by noticing
that $\liminf_{k\to\infty} \| \Psi_{n_k} 
\|_{2\mu+2}^{2\mu+2} = 0$,  together with $\| \Psi_{n_k}
\|_{H^1}$ bounded  and Gagliardo-Nirenberg inequality,  would imply
$\liminf_{k\to\infty} |\Psi_{n_k}(0,1)| =0$ and contradict inequality
\eqref{e:floor}. Since 
also for $0<\tau<m$  we get a contradiction, cf. inequalities
\eqref{e:black-1} and \eqref{e:black-2}, it must be $\tau=m$. 

\n
Now we prove that for $m<m^\ast$ the minimizing sequence is not {\em runaway}. Here the limitation on the mass plays a role for the first time.
By absurd suppose that $\Psi_n $ is {\em runaway}. Then we have that $\psi_{i,n} (0) \to 0$ by Lem. \ref{l:cc} and this implies 
\beq
\label{info}
\lim_{n\to \infty} E[\Psi_n] - E^0 [\Psi_n] = 0
\eeq
where $E^0$ is the energy functional corresponding to the Kirchhoff condition in the vertex, see Eq. \eqref{e:nrgK}. By equality \eqref{info} it must be 
\beq
\label{runnabs}
-\nu \geqslant \inf_{\stackrel{\Psi\in \EE}{ M[\Psi]=m,\; \Psi\neq 0 }} E^0 [\Psi].
\eeq
We shall provide a lower bound of $\inf E^0 [\Psi]$ by means of the rearrangements and then, by a trial function, we show that
\eqref{runnabs} is false giving an absurd.

\n
Let $\Psi^\ast$ be the rearranged function of $\Psi$. By Prop. \ref{lp} and \ref{polya} we have
\[
\| \Psi \| = \| \Psi^\ast \| \qquad \qquad
\| \Psi \|_{2\mu + 2}  = \| \Psi^\ast \|_{2\mu + 2} 
\]
and 
\[
\| \Psi^\prime \|^2 \geqslant \f{4}{N^2} \| \Psi^{\ast \prime}  \|^2 .
\]
Therefore, for a non trivial $\Psi$ such that $\Psi\in \EE$ and $
M[\Psi]=m$, we see that  $\Psi^\ast \in \EE$ due to its symmetry, $
M[\Psi^\ast]=m$ 
and
\be
E^0 [\Psi] \geqslant  \f{4}{N^2} \f{1}{2} \| \Psi^{\ast \prime}  \|^2
-  \frac{1}{2\mu + 2}   \| \Psi^\ast \|_{2\mu + 2}^{2\mu + 2}.  
\ee
Since rearrangements maintain  the mass constraint,  the previous inequality implies
\be
\inf_{\substack{\Psi\in \EE\\ M[\Psi]=m}} E^0 [\Psi] \geqslant 
\inf_{
\substack{\Psi\in \EE,\; M[\Psi]=m,\\  \Psi \text{ symmetric}}}
\f{4}{N^2} \f{1}{2} \| \Psi^{ \prime}  \|^2 - \f{1}{2\mu + 2} \| \Psi
\|_{2\mu + 2}. 
\ee
Taking into account the symmetry requirement this last problem reduces to $N$ copies of a problem on the half line
\[
\inf_{
\substack{\Psi\in \EE,\; M[\Psi]=m,\\  \Psi \text{ symmetric}}} \f{4}{N^2} \f{1}{2} \| \Psi^{ \prime}  \|^2 - \f{1}{2\mu + 2} \| \Psi \|_{2\mu + 2}^{2\mu + 2} 
 =  
N
\inf_{\substack{\psi\in H^1 (\erre^+)\\ M_{\erre^+}[\psi]=m/N} }
 \f{4}{N^2} \f{1}{2} \| \psi^{ \prime}  \|^2_{L^2 (\erre^+)} - \f{1}{2\mu + 2} \| \psi \|_{L^{2\mu + 2}(\erre^+)}^{2\mu + 2}. 
\]
It is convenient to rescale the problem by means of the unitary transform $\psi(\cdot ) \mapsto \la^{1/2} \psi( \la \cdot)$. In this way we have to 
minimize the functional
\[
\f{4}{N^2} \f{\la^2}{2} \| \psi^{ \prime}  \|^2_{L^2 (\erre^+)} - \f{\la^\mu}{2\mu + 2} \| \psi \|_{L^{2\mu + 2}(\erre^+)}^{2\mu + 2}.
\]
Choosing $\la$ such that $\f{4}{N^2} \la^2 = \la^\mu$  we  reconstruct the structure of $E_{\erre^+}$ and arrive at the following inequality
\[
\inf_{\substack{\Psi\in \EE\\ M[\Psi]=m}} E^0 [\Psi] \geqslant
 N  \lf( \f{N}{2} \ri)^{ \f{2\mu}{2-\mu} }
\inf_{\substack{\psi\in H^1 (\erre^+) \\ M_{\erre^+}[\psi]=m/N }}
E_{\erre^+ }[\psi]
% \f{1}{2} \| \psi^{ \prime}  \|^2_{L^2 (\erre^+)} - \f{1}{2\mu + 2} \| \psi \|_{L^{2\mu + 2}(\erre^+)} 
\]
which is a minimization problem for unperturbed energy on the half line.  Recalling that the solution of the constrained energy minimization problem on the half-line is given by 
the half soliton with frequency $\tilde \ome$ such that $M_{\RE^+}[\phi_{\tilde \ome}]=m/N$ we obtain  
\beq \label{inferior}
\inf_{\substack{\Psi\in \EE\\ M[\Psi]=m}} E^0 [\Psi] \geqslant
- \f N 2  \lf( \f{N}{2} \ri)^{ \f{2\mu}{2-\mu} } \;  \omestar^{ \f 1 \mu + \f 1 2}\,
\f{(\mu+1)^{\f 1 \mu }  }{\mu}   \f{2-\mu}{2+\mu}   \int_0^1 (1-t^2)^{\f 1 \mu -1} dt \; 
\eeq
with $\omestar$  defined by
\[
 \frac{m}{N}  =  \f{(\mu+1)^{\f 1 \mu }  }{\mu} \omestar^{ \f 1 \mu - \f 1 2} \int_0^1 (1-t^2)^{\f 1 \mu -1} dt ,
\]
where we used identities \eqref{MR} and \eqref{ER}.

\n
We can write the r.h.s. in a more compact way, showing also that it
does not actually depend on $N$. Let $\omega_\RE$ be the frequency of
a soliton of mass $m$, by Eq. \eqref{MR}, one has  
\[
m =  2\f{(\mu+1)^{\f 1 \mu }  }{\mu} \ome_\RE^{ \f 1 \mu - \f 1 2} \int_0^1 (1-t^2)^{\f 1 \mu -1} dt,
\]
from which it follows that 
\beq \label{rapporto}
\frac{\ome_\erre}{\omestar} = \lf( \f N 2 \ri)^{ \f{2\mu}{2-\mu} } .
\eeq
Taking into account  \eqref{inferior} and \eqref{rapporto} we have
\beq \label{infinal}
\inf_{\substack{\Psi\in \EE\\ M[\Psi]=m}} E^0 [\Psi] \geqslant  -  \ome_{\erre}^{ \f 1 \mu + \f 1 2}\,
\f{(\mu+1)^{\f 1 \mu }  }{\mu}   \f{2-\mu}{2+\mu}   \int_0^1 (1-t^2)^{\f 1 \mu -1} dt = 
- \frac12 \f{2-\mu}{2+\mu}  \; \ome_\erre \, m \; .
\eeq
This is the lower bound we were interested in. Notice that the
r.h.s. coincides with the energy of a soliton on the line with mass
$m$. 
%For later convenience we rewrite this energy in a more compact form using \eqref{retta}.
%\beq \label{infinal2}
%\inf_{\substack{\Psi\in \EE\\ M[\Psi]=m}} E_0 [\Psi] \geqslant  -  
%- \f 1 2    \f{2-\mu}{2+\mu}  \; \ome_\erre \, m \;
%\eeq

\n
Now we compute the energy  functional $E$ on a trial function.  As trial function we choose the $N$-tail state $\Psi_{\ome,0}$. First we fix the frequency $\ome=\ome_0$, where $\ome_0$ is such that $M[\Psi_{\ome_0,0}]=m$. By  Eq. \eqref{formula1} we get 
\beq \label{massacode}
M[\Psi_{\ome,0}] = N \f{(\mu+1)^{\f 1 \mu }  }{\mu}  \ome^{ \f 1 \mu - \f 1 2}  \int_{ \frac{|\al|}{ N \sqrt{\ome}}}^1 (1-t^2)^{\f 1 \mu -1} dt.
\eeq
The r.h.s. of \eqref{massacode} as a function of $\ome$ defined on the domain $[\al^2 /N^2 , \infty)$ is positive, increasing and the range is $[0,\infty)$ in the subcritical case
while in the critical case the range is $[0, \frac{\pi \sqrt{3}N}{4})$. See also Section \ref{sec5}. 
%{\bf ***}
Therefore the equation 
\[
m=N \f{(\mu+1)^{\f 1 \mu }  }{\mu}  \ome^{ \f 1 \mu - \f 1 2}  \int_{ \frac{|\al|}{ N \sqrt{\ome}}}^1 (1-t^2)^{\f 1 \mu -1} dt
\]
has a unique solution $\ome_0$ for every $m>0$ such that  $\ome_0 >
\al^2 /N^2$ . A straightforward calculation based on formulas  \eqref{formula1} - \eqref{formula3} gives 
\begin{align}
E[\Psi_{\ome_0,0}] =& -\ome_0 \frac{m}{2} +\frac{\mu}{2\mu+2}\|\Psi_{\ome,0}\|_{2\mu+2}^{2\mu+2}
\nonumber
\\=& - \frac12 \f{2-\mu}{2+\mu}  \; \ome_0 \, m - \f 1 2  \f{(\mu+1)^{\f 1 \mu }  }{\mu+2 } \mu |\al | \lf( \ome_0 - \f{\al^2}{N^2} \ri)^{\f{1}{\mu} }
 \label{energiacode}
\end{align}
Now we prove that, if $m<m^\ast$,  then 
\beq \label{preabs}
\inf_{\substack{\Psi\in \EE\\ M[\Psi]=m}} E^0 [\Psi] > E[\Psi_{\ome_0,0}]
\eeq 
Due to \eqref{infinal} and \eqref{energiacode}  it is sufficient to show that
\[
\ome_0>\ome_\erre.
\]
Notice that the condition $m<m^\ast$ is equivalent, see \eqref{mstar}, to
\[
\ome_\erre < \f{\al^2}{N^2}.
\]
Since we have $\ome_\erre < \f{\al^2}{N^2} < \ome_0$, then \eqref{preabs} is proved. 
This is absurd since by \eqref{runnabs} we have
\[
E[\Psi_{\ome_0,0}] \geqslant -\nu \geqslant \inf_{\substack{\Psi\in \EE\\ M[\Psi]=m}} E^0 [\Psi] > E[\Psi_{\ome_0,0}].
\]
Then $\Psi_n$ is not {\em runaway} and therefore it is convergent, up
to subsequences, to $\hat \Psi$ in $L^p (\GG)$ for $p  \geq 2$. In particular, $M[\hat \Psi] = m$.
Moreover taking into account also the weak lower continuity of the $H^1$ norm we have
\[
E[ \hat \Psi] \leq \lim_{n\to \infty} E[ \Psi_n] = -\nu
\]
which implies that $E[ \hat \Psi] = -\nu$. Since $E[ \hat \Psi] = \lim_{n\to \infty}  E[ \Psi_n]$ then $\| \hat \Psi ' \| = \lim_{n\to \infty} \|  \Psi_n ' \|$ and we
have proved that $\Phi_n \to \hat \Psi$ in $H^1$.
\end{proof}

\begin{remark}
The condition $m<m^\ast$ has the advantage to be explicit, however we
stress that it is not optimal. Indeed, for any $m$ such that
\eqref{preabs} is satisfied, the proof given holds true. By careful
inspection of  \eqref{preabs} this is true for $m=m^\ast$ 
and by continuity also for some $m>m^\ast$.  
%{\bf ***}
\end{remark}
%%%%%%
%SECTION
%%%%%%
\section{Energy ordering of the stationary states} \label{sec5}

In this section we study the energy ordering of the stationary states
for fixed mass in critical and subcritical regime.  
In both cases we prove that the energy of the stationary states at fixed mass is increasing in the  number of bumps. 
Therefore, among the stationary states with equal mass,  the $N$-tail
state has minimal energy, see Th. \ref{t:min}. 
In the critical case a new restriction on $m$ appears. First we analyze the subcritical case.

%%%%%%%%
%SUBSECTION
%%%%%%%%
\subsection{Energy ordering of the stationary states: subcritical nonlinearity}

%We start by analyzing the behavior of the frequencies of the stationary states  for fixed mass as functions of the number of bumps. We show that the ordering of the frequencies changes when $\mu$ crosses $1$.

We consider as usual the case $\alpha<0$ only, then we set $\alpha = - |\alpha|$. 
We define the functions $M_j(\ome) = M[\Psi_{\ome,j}]$.  A straightforward calculation gives 
\begin{align}
M_j(\ome) & = \frac{(\mu+1)^{\frac1\mu}}{\mu} \omega^{\frac1\mu-\frac12}  \bigg[-(N-2j)\int_0^{\frac{|\al|}{(N-2j)\omega^{\frac12}}} (1-t^2)^{\frac1\mu -1} dt + N  I\bigg]
%\label{e:mome}
\nonumber
\\
 & =  \frac{(\mu+1)^{\frac1\mu}}{\mu} \omega^{\frac1\mu-\frac12}  \bigg[(N-2j) \int_{\frac{|\al|}{(N-2j)\omega^{\frac12}}}^1 (1-t^2)^{\frac1\mu -1} dt + 2j  I \bigg]
\label{e:mome2}
\end{align}
where 
\[
I=\int_0^1(1-t^2)^{\frac1\mu-1} dt
\]
 We recall that $\Psi_{\ome,j}$ is defined for $\omega\in\left(\frac{|\al|^2}{(N-2j)^2},\infty\right)$. Notice that the stationary states, apart from the $N$-tail state, have a minimal
mass, that is the range of the functions $M_j$, denoted as $\Ran M_j$, is separated from zero. In fact,  we have that
\[
\Ran M_j= M_j(\frac{|\al|^2}{(N-2j)^2},\infty)=\left[2 jI
  \frac{(\mu+1)^\frac{1}{\mu}}{\mu}
  \left(\frac{|\al|}{(N-2j)}\right)^{\frac{2-\mu}{\mu}} ,\infty
  \right) 
\]

\n
First we compare the frequency of the stationary states on the manifold $M[\Psi]=m$.
%%%%%
%LEMMA
%%%%%
\begin{lemma}[Frequency ordering]
\label{l:freq}
 Let $0<\mu<2$ and take $\Psi_{\ome,j}$ defined by \eqref{states1} and \eqref{states2}. Assume that
%then for any $j=0,...,[(N-1)/2]$  if and only if 
 \begin{equation}
 \label{e:minmass}
 m \geq 2 j \left(\frac{|\al|}{(N-2j)}\right)^{\frac{2-\mu}{\mu}}
 \frac{(\mu+1)^\frac{1}{\mu}}{\mu} \int_0^1(1-t^2)^{\frac1\mu-1} dt
 \, , 
 \end{equation}
then  there exists $\ome_j$ such that $M[\Psi_{\ome_j,j}]=m$.
Moreover, assume that condition \eqref{e:minmass} is satisfied for
$j+1$ (and therefore for $j$). The following possibilities hold:

- if $0<\mu<1$ then 
%$m$ is such that both $\ome_j$  and $\ome_{j+1}$ exist then 
\begin{equation}
\label{e:omeord1}
 \ome_{j+1}<\ome_j \,;
\end{equation}

- if $\mu=1$, then $\ome_j$  is independent of $j$ and 
\begin{equation}
\label{e:omemu=1}
 \ome_j\equiv \ome^* =\frac{ (m + 2|\al|)^2 }{4 N^2} \,; 
 \end{equation}

- if $1<\mu<2$,  then 
\begin{equation}
\label{e:omeord2}
 \ome_{j+1}>\ome_j\,.
\end{equation}
\end{lemma}
%%%%%
%PROOF
%%%%%
\begin{proof}
The frequency  $\ome_j$ is the solution to the equation $m = M_j(\ome_j)$, 
then for each $j$ the equation $M_j(\ome) = m$ has solution only if $m \geq 2 jI \frac{(\mu+1)^\frac{1}{\mu}}{\mu} \left(\frac{|\al|}{(N-2j)}\right)^{\frac{2-\mu}{\mu}} $, which proves the first part of the lemma. 
Next we note that the functions $M_j$ are strictly
increasing. Moreover,    for any $\ome \geq \frac{|\al|^2}{(N-2(j+1))^2} $ 
\[
M_{j+1}(\ome) - M_j(\ome)  = - 
\frac{(\mu+1)^{\frac1\mu}}{\mu} \omega^{\frac1\mu-\frac12}  
\left[
\int_0^{\frac{|\al|}{\sqrt\omega}} \left(1-\frac{t^2}{(N-2(j+1))^2}\right)^{\frac1\mu -1} 
-
\left(1-\frac{t^2}{(N-2j)^2}\right)^{\frac1\mu -1} dt 
\right].
\]
Since the function 
\[
\left(1-\frac{t^2}{(N-2(j+1))^2}\right)^{\frac1\mu -1} 
-
\left(1-\frac{t^2}{(N-2j)^2}\right)^{\frac1\mu -1}
\]
is negative for $0<\mu<1$ and positive for $1<\mu<2$, one has that
$M_{j+1}(\ome) > M_j(\ome) $ for $0<\mu<1$ and $M_{j+1}(\ome) <
M_j(\ome) $ for $1<\mu<2$. Together with the fact that $M_j$ are
strictly increasing functions, this provides the ordering \eqref{e:omeord1}
and \eqref{e:omeord2}.  

Formula \eqref{e:omemu=1} is obtained by setting $\mu=1$ into equation
$M_j(\ome) = m$ and through a straightforward calculation, see also
\cite{[ACFN4]}.   
\end{proof}

%%%%%%%%%%%
%LEMMA
%%%%%%%%%%%
\begin{lemma}[Energy ordering]
\label{l:nrgord}
 Let $0<\mu<2$, and $m$ be such that condition \eqref{e:minmass} is satisfied for $ j+1$. 
%there exist $\ome_j$ and $\ome_{j+1}$ such that $M[\Psi_{\ome_{j+1},j+1}]= M[\Psi_{\ome_{j},j}]=m$, see   condition \eqref{e:minmass}. 
Then
\begin{equation}
\label{e:nrgord}
E[\Psi_{\ome_j,j}] < E[\Psi_{\ome_{j+1},j+1}] \,.
\end{equation}
\end{lemma}
%%%%%%%%%%%
%PROOF
%%%%%%%%%%%
\begin{proof}
After some straightforward  calculation using \eqref{formula1} \eqref{formula2} and \eqref{formula3}, one gets the formula
\begin{equation}
\label{e:Ej}
 E[\Psi_{\ome_j,j}] = -\frac{1}{2(\mu+2)}\left[m\ome_j(2-\mu) + |\al| \mu(\mu+1)^{\frac1\mu} \left(\ome_j - \frac{ |\al|^2}{(N-2j)^2}\right)^{\frac1\mu} \right] \,.
\end{equation}
Let us set 
\[
\Delta_j = E[\Psi_{\ome_{j+1},j+1}]  -  E[\Psi_{\ome_j,j}] \,.
\]
We aim at proving that $\Delta_j >0$. One has 
\begin{equation}
\label{e:Dej}
\Delta_j =  -\frac{m(2-\mu) }{2(\mu+2)} (\ome_{j+1}-\ome_j) 
 -\frac{|\al| \mu(\mu+1)^{\frac1\mu}}{2(\mu+2)} \left[\left(\ome_{j+1} - \frac{ |\al|^2}{(N-2(j+1))^2}\right)^{\frac1\mu}  - \left(\ome_j - \frac{ |\al|^2}{(N-2j)^2}\right)^{\frac1\mu}\right]\,.
\end{equation}
Let us analyze separately the cases $0<\mu\leq1$ and $1<\mu<2$. 

We start with the case $0<\mu\leq1$, the easiest one. By
Lem. \ref{l:freq} one has that $(\ome_{j+1}-\ome_j) <0$ (equality
holds only for $\mu=1$). From which it also follows that  
\[
\left(\ome_{j+1} - \frac{ |\al|^2}{(N-2(j+1))^2}\right)^{\frac1\mu}  - \left(\ome_j - \frac{ |\al|^2}{(N-2j)^2}\right)^{\frac1\mu} < 0\,.
\]
Noting that  $\Delta_j$ is the sum of two positive terms, we obtain \eqref{e:nrgord} for  $0<\mu\leq1$. 

The case $1<\mu<2$ is more difficult. To prove \eqref{e:nrgord} we start from equation \eqref{e:mome2} and recall that the frequency $\ome_j$ satisfies the equality $m = M_j[\ome_j] $, i.e.
\[
m =\frac{(\mu+1)^{\frac1\mu}}{\mu} \omega_j^{\frac1\mu-\frac12} \bigg[ (N-2j) \int_{\frac{|\al|}{(N-2j)\omega_j^{\frac12}}}^1 (1-t^2)^{\frac1\mu -1} dt + 2j  I \bigg]\,.  
\]
Taking the left and right derivative with respect to $m$, after some
straightforward calculation we 
obtain  
\[
\f{d}{dm}\ome_j = 2\mu \ome_j \bigg[ m(2-\mu) + |\al| (\mu+1)^{\frac1\mu} \left( \ome_j - \frac{|\al|^2}{(N-2j)^2}
\right)^{\frac1\mu -1} \bigg]^{-1} \,. 
\]
%here $\ome_j' = \frac{d}{dm}\ome_j$. 
Then, taking the derivative of $E[\Psi_{\ome_j,j}]$ in equation
\eqref{e:Ej} and using the last identity, we obtain: 
\begin{align*}
\frac{d }{dm}E[\Psi_{\ome_j,j}] = &  -\frac{2-\mu}{2(2+\mu)}\ome_j  
 -\frac{1}{2(2+\mu)}\left[m(2-\mu) + |\al| (\mu+1)^{\frac1\mu} \left(\ome_j - \frac{ |\al|^2}{(N-2j)^2}\right)^{\frac1\mu-1} \right]\ome_j'  \\ 
= & -\frac{\ome_j}{2}\,.
\end{align*}
Together with \eqref{e:omeord2}, latter formula implies that for
$1<\mu<2$, $\Delta_j$ is  decreasing in $m$: 
\[
\frac{d}{dm} \Delta_j = -  \frac12(\ome_{j+1}-\ome_j) < 0 \,.
\]
Then to prove \eqref{e:nrgord} it is enough to prove that $\Delta_j > 0$ for $m \to \infty$. 
To prove the latter statement we start by equation \eqref{e:mome2}
and notice that  $\ome_j\to\infty$ as  $m\to\infty $, moreover  
by the expansion $\int_x^1(1-t^2)^{\frac1\mu-1} = I -x +\frac13
\left(\frac1\mu -1\right)x^3 + O(x^5)$,  we obtain
\begin{equation}
\label{e:mome3}
 \frac{m}{NC} = \omega_j^{\frac1\mu-\frac12} \bigg[1 -\frac{|\al|}{NI} \ome_j^{-\frac12} + \frac1{3NI}\left(\frac1\mu-1\right) \frac{|\al|^3}{(N-2j)^2} \ome_j^{-\frac32} +O(\ome^{-\frac52}) \bigg]\,,
\end{equation}
where $C= \frac{(\mu+1)^\frac1\mu}{\mu} I$. 
For $m\to\infty$, $\ome_j$ has the following expansion:
\begin{equation}
\label{e:expome}
\ome_j = \lf( \f{m}{N C}\ri)^{ \frac{2\mu}{2-\mu} }\lf[ 1 + a_j  \lf( \f{m}{N C}\ri)^{ -\frac{\mu}{2-\mu} }
+ b_j \lf( \f{m}{N C}\ri)^{ -\frac{2\mu}{2-\mu} } + c_j \lf( \f{m}{N C}\ri)^{ -\frac{3\mu}{2-\mu} }+ O(m^{-\frac{4\mu}{2-\mu}}) \ri] \,.
\end{equation}
To compute the coefficients $a_j$, $b_j$ and $c_j$ we rewrite equation \eqref{e:mome3} in the form 
\[ 
\left(\frac{m}{NC} \right)^{\frac{2\mu}{2-\mu}}= \omega_j \bigg[1 -\frac{|\al|}{NI} \ome_j^{-\frac12} + \frac1{3NI}\left(\frac1\mu-1\right) \frac{|\al|^3}{(N-2j)^2} \ome_j^{-\frac32} +O(\ome^{-\frac52}) \bigg]^{\frac{2\mu}{2-\mu}}\,,
\]
and use formula  \eqref{e:expome} at the r.h.s.. The r.h.s. has an
expansion in powers $\left(\frac{m}{NC} \right)^{-\frac{j\mu}{2-\mu}}$
with $j=-2,-1,0,1, ...$. The condition that the terms with $j=-1,0, 1$
equal zero gives the coefficients $a_j$, $b_j$ and $c_j$. A lengthy
but straightforward calculation shows that the coefficients $a_j$ and
$b_j$ are independent of $j$. This  
is due to the fact that the first term in equation \eqref{e:mome3} does
not depend on $j$. 
More precisely, one obtains:
\[
a_j\equiv a  = \frac{2\mu}{2-\mu} \frac{|\al|}{NI}
\;;\quad
b_j \equiv b = \frac{\mu}{2-\mu} \frac{|\al|^2}{N^2 I^2}
\;;\quad
c_j =  c -  \frac{2(1-\mu)}{2-\mu} \frac{1}{3NI} \frac{|\al|^3}{(N-2j)^2},
\]
where $ c$ does not depend on $j$. The explicit expression is not
relevant since it will cancel out (see below). 
Using the expansion \eqref{e:expome} in equation \eqref{e:Dej} and taking into account the fact that the coefficients $a_j\equiv a$ and $b_j\equiv b$ do not depend on $j$ we obtain the following expansion for $\De_j$
\[
\begin{aligned}
\De_j = &
  -\frac{(2-\mu) }{2(2+\mu)} \left(NC\right)^{1+\frac{2\mu-2}{2-\mu}} (c_{j+1}-c_j) m^{-\frac{2\mu-2}{2-\mu}} \\
  & +\frac{|\al| (\mu+1)^{\frac1\mu}}{2(\mu+2)}  \left(NC\right)^{\frac{2\mu-2}{2-\mu}}   \left(\frac{|\al|^2}{(N-2(j+1))^2} - \frac{|\al|^2}{(N-2j)^2}\right) m^{-\frac{2\mu-2}{2-\mu}}
  +O\left(m^{-\frac{3\mu-2}{2-\mu}}\right)
\\
= & 
  \frac{(\mu+1)^{\frac1\mu}}{2(\mu+2)} |\al|^3 \left(NC\right)^{\frac{2\mu-2}{2-\mu}} 
   \left(\frac{|\al|^2}{(N-2(j+1))^2} - \frac{|\al|^2}{(N-2j)^2}\right) \left(\frac{2}{3\mu}+\frac13\right)
  m^{-\frac{2\mu-2}{2-\mu}} 
  +O\left(m^{-\frac{3\mu-2}{2-\mu}}\right)
\end{aligned}
\] 
where in the latter equality we used the definition of $c_j$ and the fact that $I= \frac\mu{(\mu+1)^\frac1\mu}C$. The latter equality shows that for $m$ large enough $\Delta_j$ is positive 
for any $0<\mu<2$, 
and the proof of the lemma is concluded.
\end{proof}

Lem. \ref{l:nrgord} shows that among the stationary states on the manifold $M[\Psi]=m$ the $N$-tail state has minimum energy and therefore for $0<\mu<2$ 
the proof of Th. \ref{t:min} immediately follows.
%$\hat \Psi$ must be one of the stationary states.
%%%%%%
%REMARK
%%%%%%
\begin{remark}
For $\mu=1$ the energy spectrum at fixed mass  can be explicitly computed:
\[
E[\Psi_{\ome,j}] = -\frac{N}{3} \ome^{\frac32} + \frac13 \frac{|\al|^3}{(2j - N)^2}\,.
\]
Taking into account the mass constraint we have
\[
E[\Psi_{\ome^\ast,j}] = -\frac{1}{24} \frac{ (m+2|\al|)^3}{N^2} + \frac13 \frac{|\al|^3}{(2j - N)^2}\,.
\]
The energy of the ground state is given by
\[
E[\Psi_{\ome^\ast,0}] = -\frac{1}{24 N^2} m (m^2 + 6m|\al| + 12 |\al|^2 )
\]
\end{remark}
\begin{remark}
Notice that the manifold $M[\Psi]=m$ for $m<m^\ast$ may not contain all the stationary states, due to the fact that their masses have a lower bound, as discussed above.
The $N$-tail state always belongs to the constraint manifold since its mass has no lower bound. Since $m^\ast$ actually depends on $\al$ , by inspection it turns out that for small $|\al|$ the
constraint manifold contains only the N-tail state while for large $|\al|$ all the stationary states belong to the constraint manifold, i.e. the equation $M_j (\ome) = m$
defines the frequency $\ome_j$.  As a matter of fact, for the proof of our theorems we could fix $m$ and require $\al$ to be sufficiently negative. 
Analogous remarks also apply to the critical case. 
%{\bf ***}
\end{remark}

%%%%%%%%
%SUBSECTION
%%%%%%%%
\subsection{Energy ordering of the stationary states: critical nonlinearity}
In this section we study the energy ordering of the stationary states for fixed mass and $\mu=2$. 

%Also in this section we consider only the case $\alpha<0$, then we set $\alpha = - |\alpha|$. 
%%%%%
%LEMMA
%%%%%

\n
In the critical case the mass functions can be explicitly computed and we have.
\begin{align*}
M_j(\ome) & = \frac{\sqrt 3}{2}  \bigg[-(N-2j)\int_0^{\frac{|\al|}{(N-2j)\omega^{\frac12}}} (1-t^2)^{-\frac12} dt + N  I\bigg]
\\
%\label{e:momemu5}
&= \frac{\sqrt 3}{2}  \bigg[-(N-2j)  \arcsin\left(\frac{|\al|}{(N-2j)\omega^{\frac12}}\right) + \frac{N\pi}{2}  \bigg]
\end{align*}
where we used the fact that  $I=\int_0^1(1-t^2)^{-\frac12} dt = \pi/2$. We note that
\[
\Ran M_j = \left[ j \frac{\pi \sqrt3}{2}  , \frac{N}{2} \frac{\pi \sqrt3}{2}    \right)
\]
In the critical case all the mass functions are bounded from above,
therefore for large $m$ the frequencies $\ome_j$ are not defined. 
This is the reason of the further mass limitation appearing in Ths. \ref{t:prob1} and \ref{t:min}.

\begin{lemma}[Frequency ordering $(\mu=2)$]
\label{l:freqcrit}
 Let $\mu=2$ and take $\Psi_{\ome,j}$ defined by \eqref{states1} and
 \eqref{states2}.  
Assume that
\begin{equation}
\label{e:co}
 j\frac{\pi\sqrt3}{2}  \leq m < \frac{N}{2} \frac{\pi \sqrt3}{2} ,  
\end{equation}
then  there exists $\ome_j$ such that $M[\Psi_{\ome_j,j}]=m$.
Moreover, if $m$ is such that \eqref{e:co} is satisfied for $j+1$
(therefore also for $j$) then:  
%$\ome_j$  and $\ome_{j+1}$ exist then 
\begin{equation}
\label{e:omeord3}
 \ome_{j+1}>\ome_j\,.
\end{equation}
\end{lemma}
%%%%%
%PROOF
%%%%%
\begin{proof}
 We recall that
 $\omega\in\left(\frac{|\al|^2}{(N-2j)^2},\infty\right)$,  the
 frequency  $\ome_j$ is the solution to the equation $m =
 M_j(\ome_j)$. 
then for each $j$ the equation $M_j(\ome) = m$ has solution if and
only if  $ j\frac{\pi\sqrt3}{2}  \leq m < \frac{N}{2} \frac{\pi
  \sqrt3}{2}   $, which proves the first part of the theorem.   

To prove the second part of the theorem we solve the equation $m = M_j(\ome_j)$  for $\ome_j$ and obtain 
\[
 \ome_j = \frac{|\al|^2}{(N-2j)^2 \sin\left( \frac\pi2 \frac{N-\frac{4m}{\pi\sqrt3 }}{N-2j}\right)^2}\,.
\]
And the ordering \eqref{e:omeord3} is proved by noticing that the function 
\[
 f(x) = \frac{|\al|}{(N-2x) \sin\left( \frac\pi2 \frac{N-\frac{4m}{\pi\sqrt3 }}{N-2x}\right)}
\]
is increasing whenever the argument of the $\sin$ is in
$(0,\pi/2)$. This is our case because of the constraint
\eqref{e:co}, as it is easily seen by taking the derivative with
respect to $x$ 
\[
 f'(x) = \frac{2|\al|}{(N-2x)^2 \sin\left( \frac\pi2 \frac{N-\frac{4m}{\pi\sqrt3 }}{N-2x}\right)^2} 
\left(\sin y - y \cos y\right)\bigg|_{y = \frac\pi2 \frac{N-\frac{4m}{\pi\sqrt3 }}{N-2x}}\,,
\]
then $f'(x) > 0$ by the inequality $\sin y - y \cos y>0$ which holds true for any $0<y<\pi/2$.
\end{proof}

%%%%%%%%%%%
%LEMMA
%%%%%%%%%%%
\begin{lemma}[Energy ordering $(\mu=2)$]
 Let $\mu=2$ and assume that \eqref{e:co} is satisfied for $j+1$. Then,
%Then if  $m$ is such that both $\ome_j$  and $\ome_{j+1}$ exist  one has
\[
%\label{e:nrgordcrit}
E[\Psi_{\ome_j,j}] < E[\Psi_{\ome_{j+1},j+1}] \,.
\]
\end{lemma}
%%%%%%%%%%%
%PROOF
%%%%%%%%%%%
\begin{proof}
After some straightforward  calculation one gets the formula
\begin{align*}
 E[\Psi_{\ome_j,j}] = & -\frac{|\al| \sqrt3 }{4}  \left(\ome_j - \frac{ |\al|^2}{(N-2j)^2}\right)^{\frac12}   \\ 
= & - \frac{\sqrt3}{4} \frac{|\al|^2}{(N-2j)} \left( \frac{1}{\sin\left( \frac\pi2 \frac{N-\frac{4m}{\pi\sqrt3 }}{N-2j}\right)^2} -1 \right)^{\frac12}
\end{align*}
where we used the explicit formula for $\ome_j$.  Taking the derivative of the function 
\[
f(x) = 
- \frac{\sqrt3}{4} \frac{|\al|^2}{(N-2x)} \left( \frac{1}{\sin\left( \frac\pi2 \frac{N-\frac{4m}{\pi\sqrt3 }}{N-2x}\right)^2} -1 \right)^{\frac12}
\]
we have that 
\[
f'(x) = 
 \frac{\sqrt3}{4} \frac{|\al|^2}{(N-2x)^2} \frac{2}{\left( \frac{1}{(\sin y )^2} -1 \right)^{\frac12}}
 \frac{1}{(\sin y)^2}
\left((\sin y)^2-1 + \frac{y}{\tan y}\right)\bigg|_{y = \left( \frac\pi2 \frac{N-\frac{4m}{\pi\sqrt3 }}{N-2x}\right)}
\]
and the energy ordering is a consequence of the fact that $f'(x)>0$,
which follows from the inequality $(\sin y)^2 -1 + \frac{y}{\tan y} >
0 $ and is true for any $0<y<1$. 
\end{proof}

\n
This ends the proof of Th. \ref{t:min}.

\newcommand{\etalchar}[1]{$^{#1}$}
\providecommand{\bysame}{\leavevmode\hbox to3em{\hrulefill}\thinspace}
\providecommand{\MR}{\relax\ifhmode\unskip\space\fi MR }
\providecommand{\MRhref}[2]{%
  \href{http://www.ams.org/mathscinet-getitem?mr=#1}{#2}
}
\providecommand{\href}[2]{#2}

%\bibliographystyle{myamsalpha}
%\bibliographystyle{myamsplain}
%\bibliography{conc-comp}

\end{document}